\documentclass[journal]{IEEEtran}
\usepackage{amsmath,amssymb,amsfonts,bm}
\usepackage{amsthm}
\usepackage{booktabs}
\usepackage{graphicx}
\usepackage{cite}
\usepackage[colorlinks=false,hidelinks]{hyperref}
\usepackage{orcidlink}
\usepackage{xcolor}
\usepackage{algorithm}
\usepackage{algorithmic}
\newcommand{\im}{\jmath}                        
\newcommand{\Imag}{\operatorname{Im}}
\newcommand{\Real}{\operatorname{Re}}

\newcommand{\diag}{\operatorname{diag}}
\newcommand{\dd}{\mathrm{d}}

\newcommand{\br}{\mathbf{r}}
\newcommand{\bu}{\bar{\mathbf{u}}}
\newcommand{\bk}{\bar{\mathbf{k}}}
\newcommand{\bp}{\mathbf{p}}
\newcommand{\bw}{\mathbf{w}}
\newcommand{\bh}{\mathbf{h}}
\newcommand{\bG}{\mathbf{G}}
\newcommand{\bI}{\mathbf{I}}
\newcommand{\bC}{\mathbf{C}}
\newcommand{\bB}{\mathbf{B}}
\newcommand{\bR}{\mathbf{R}}
\newcommand{\bX}{\mathbf{X}}
\newcommand{\bZ}{\mathbf{Z}}
\newcommand{\bj}{\mathbf{j}}
\newcommand{\bq}{\mathbf{q}}
\newcommand{\surf}{\mathcal{S}}
\newcommand{\Utt}{\mathcal{U}}
\newcommand{\kerZ}{\boldsymbol{\mathcal{Z}}}
\newcommand{\kerR}{\boldsymbol{\mathcal{R}}}
\newcommand{\kerX}{\boldsymbol{\mathcal{X}}}
\newcommand{\bsb}{\bar{\mathbf{s}}}             
\newcommand{\brb}{\bar{\mathbf{r}}}             
\newcommand{\bpb}{\bar{\mathbf{p}}}             
\newcommand{\btau}{\bar{\bm{\tau}}}             
\newcommand{\ba}{\mathbf{a}}
\newcommand{\bA}{\mathbf{A}}
\newcommand{\Herm}{\operatorname{Herm}}

\newcommand{\be}{\mathbf{e}}
\newcommand{\bE}{\mathbf{E}}
\newcommand{\bg}{\mathbf{g}}
\newcommand{\bv}{\mathbf{v}}
\newcommand{\bM}{\mathbf{M}}                    
\newcommand{\bT}{\mathbf{T}}
\newcommand{\bW}{\mathbf{W}}
\newcommand{\bY}{\mathbf{Y}}
\newcommand{\biota}{\bm{\iota}}
\newcommand{\bLam}{\bm{\Lambda}}
\newcommand{\bPhi}{\bm{\Phi}}                   
\newcommand{\bPsi}{\bm{\Psi}}                   
\newcommand{\bP}{\mathbf{P}}                    
\newcommand{\bF}{\mathbf{F}}                    

\newcommand{\bU}{\mathbf{U}}                    
\newcommand{\bQ}{\mathbf{Q}}
\newcommand{\bXi}{\bm{\Xi}}
\newcommand{\bTheta}{\bm{\Theta}}
\newcommand{\bphi}{\bm{\varphi}}
\newcommand{\netset}{\mathcal{N}}               
\newcommand{\plane}{\mathcal{P}}
\newcommand{\lagr}{\mathcal{L}}
\newcommand{\fieldE}{\bm{\mathcal{E}}}
\newcommand{\blkdiag}{\operatorname{blkdiag}}

\newcommand{\ind}[1]{\mathbf{1}\{#1\}}
\newcommand{\Iset}{\mathcal{I}}                 
\newif\ifwithappendix
\withappendixtrue      
\newcommand{\apx}[2]{Appendix~\ifwithappendix\ref{#1}\else#2 of the supplementary material\fi}
\newtheorem{theorem}{Theorem}
\newtheorem{proposition}{Proposition}
\newtheorem{lemma}{Lemma}
\newtheorem{corollary}{Corollary}
\newtheorem{remark}{Remark}
\begin{document}
\title{Stacked Fluid Metasurfaces: Mutual-Coupling-Aware Modeling and Optimization}
\author{Giovanni Iacovelli$^{\orcidlink{0000-0002-3551-4584}}$,~\IEEEmembership{Member,~IEEE}, Chandan~Kumar~Sheemar$^{\orcidlink{0000-0003-1676-5983}}$,~\IEEEmembership{Member,~IEEE}, and \\Symeon Chatzinotas$^{\orcidlink{0000-0001-5122-0001}}$,~\IEEEmembership{Fellow,~IEEE}%
\thanks{The authors are with the Signal Processing and Communications (SIGCOM) Research Group at Interdisciplinary Centre for Security, Reliability and Trust (SnT), University of Luxembourg, 1855 Luxembourg City, Luxembourg.}}

\maketitle

\begin{abstract}
Stacked intelligent metasurfaces process the transmitted field layer by layer, and fluid antennas make the position of every radiator a design variable. Combined, they pack radiators at sub-wavelength spacings within and across layers, where mutual coupling governs the physics that current models omit. This paper develops a coupling-consistent model of a fluid transmit layer illuminating a stack of passive fluid beyond-diagonal layers. The multiport impedance matrix of any port constellation is shown to be the impedance kernel sampled at the three-dimensional separations, in one closed spherical-Hankel form whose in-plane restriction is intra-layer coupling and whose axial restriction replaces the scalar Rayleigh--Sommerfeld propagator. Only the transmitter is driven: the layer currents are induced, and the familiar cascade is the single-pass limit of one matrix inversion. In the wavenumber domain the light circle still separates radiation from reaction, now with a plane-wave propagator attached, every Fourier mode sees a transmission line loaded by the layers, and an efficiency identity shows that what the stack costs in efficiency is set by its total induced-current norm, not by its layer count. Because the load is never assumed layer-diagonal, interconnections may join atoms on different sheets, and vertical networks repeated at every site are spectrally local. Alternating optimization designs precoders, positions, and loads with closed-form gradients. Numerically, fluid ports on a $\lambda/8$ grid perform within a few per cent of continuously movable ones, while a design based on the conventional cascade model can do worse than open-circuiting the stack.
\end{abstract}


\IEEEpeerreviewmaketitle

\section{Introduction}

\IEEEPARstart{T}{wo} recent ideas relocate signal processing from the baseband into the field. Stacked intelligent metasurfaces (SIMs) cascade programmable transmissive layers in front of a feed array, so that the field is shaped by propagation and re-modulation before it leaves the transmitter \cite{An2023SIM,An2023SIMmu,Sheemar2026SIMsurvey}. Movable and fluid antenna systems (MA/FAS) turn the position of a radiator into a design variable, continuously or over a dense port grid \cite{Zhu2024MA,Wong2021FAS,New2025FAStutorial}. The two combine naturally: a transmit layer of flexibly positioned ports followed by passive layers whose meta-atoms are themselves fluid and whose response is realized by a reconfigurable interconnection network, a beyond-diagonal RIS (BD-RIS) \cite{Shen2022BDRIS,Li2022BDRIS,Nerini2024graph,Li2025BDRIStutorial}. In front of a transmitter, such a layer refracts and reflects at once, the regime of STAR-RIS \cite{Xu2021STAR,Liu2021STAR,Sheemar2026STARBD}. Fluid-element SIMs \cite{Wei2026FSIM}, fluid STAR-RIS \cite{Liu2025FSTARNOMA,MESTARS2024}, liquid metasurfaces fed by fluid antennas \cite{Shen2025LIM}, and SIM-assisted FAS \cite{Papazafeiropoulos2026SIMFAS} have all appeared within the last year.

These architectures are compact in three dimensions: fluid ports sit well below half a wavelength apart, SIM layers are stacked at fractions of a wavelength, and a compact stack is the preferable one \cite{Wei2026FSIM}. There, currents on distinct radiators interact through the near field, within and across layers and back onto the feed, where the design models are least valid. The standard SIM model cascades diagonal phase masks with a scalar Rayleigh--Sommerfeld propagator \cite{An2023SIM,An2023SIMmu}. It has no intra-layer coupling, no inter-layer reflection, no feedback onto the transmitter, no polarization, and a norm-type power constraint whose physical meaning collapses once several radiators are active \cite{Iacovelli2026MC}. The fluid-element SIM \cite{Wei2026FSIM}, the fluid STAR-RIS \cite{Liu2025FSTARNOMA,MESTARS2024}, and the movable-element and fluid RIS \cite{Zhao2025MERIS,Ye2025fRIS,Zhu2025FRISsec,Zhu2026FRISIM} inherit the same abstraction, the last flagging mutual coupling as an open challenge.

Closest in intent, Xu et al. \cite{Xu2026MABDRIS} ask whether the elements of a BD-RIS should be moved, interconnected, or both, with a far-field, coupling-free model in which the two freedoms act on channel phases and never on each other. Once ports are packed below half a wavelength, moving an element changes the impedance matrix that the network loads, so position and connectivity reshape one operator, and a stack adds an axial dimension in which both act.

\subsection{Two modeling cultures, and a third dimension}

The literature that takes coupling seriously in stacked and reconfigurable surfaces is circuit-first. Nerini and Clerckx derived a physically consistent scattering model of SIM, including a BD-RIS implementation, and the assumptions that recover the cascade \cite{Nerini2024SIMBD}. Yahya \emph{et al.} recast it in transfer-scattering form, showing that coupling and inter-layer feedback can raise the sum rate \cite{Yahya2025Tparam}, and Abrardo \emph{et al.} modeled and optimized the stack as an electromagnetic collaborative object \cite{Abrardo2025SIM,Pettanice2026SIMvalid,Abrardo2026nonlinear}; the survey \cite{Sheemar2026SIMsurvey} lists tractable, physically consistent SIM models among the open challenges. Single-layer BD-RIS has coupling-aware models in scattering and impedance form \cite{Li2023BDRISMC,Nerini2024closedform,Qian2020MC}, and so does the amplifying STAR BD-RIS \cite{Sheemar2026STARBD}. In all of these the geometry is fixed and the coupling matrix comes from a full-wave solver or a per-element formula: the models are consistent, but they cannot be differentiated in the positions. Movable arrays have only just admitted coupling, showing with circuit models of elementary radiators that it can lift the half-wavelength spacing rule and raise directivity and capacity \cite{Zhu2026MAMC,Xu2026MAdir,Liao2026MAMC}.

The field-first culture of holographic MIMO and continuous apertures models coupling as an operator, the impedance kernel, with closed-form resistive and reactive parts and an asymptotically diagonal wavenumber representation \cite{Pizzo2025MC,Wang2025CAPAMC}. The two cultures were unified in \cite{Iacovelli2026MC} for a single flexible-position surface: the impedance matrix is the kernel sampled at the port separations, the light circle splits the kernel exactly into radiation and reaction, and a moving port is a constant-modulus wavenumber codeword. The low-rank structure that keeps large beyond-diagonal surfaces affordable followed in \cite{Iacovelli2026XLBD}. All of this lives in two dimensions, and a stack lives in three. The dyadic kernel of \cite{Iacovelli2026MC} already depends on the three-dimensional separation, but what it says across layers has not been used.

As Table~\ref{tab:sota} shows, physically consistent stack models have fixed ports and numerical impedances, flexible-position surfaces are coupling-free, coupling-aware movable arrays are single arrays of elementary radiators, and closed-form wavenumber models are single-layer. A coupling-consistent, closed-form, three-dimensional model of a fluid transmitter and its stack of passive fluid beyond-diagonal layers, differentiable in every position and load, has not been developed~yet.

\begin{table}[t]
\centering
\caption{State of the art on mutual coupling (MC) in stacked and flexible-position architectures.}
\label{tab:sota}
\scriptsize
\setlength{\tabcolsep}{3pt}
\begin{tabular}{@{}p{0.20\columnwidth}p{0.26\columnwidth}p{0.28\columnwidth}p{0.16\columnwidth}@{}}
\toprule
Reference & Architecture & MC model & Domain \\
\midrule
\cite{An2023SIM,An2023SIMmu,Wei2026FSIM} & SIM, fluid-element SIM & none, RS cascade & spatial \\
\cite{Nerini2024SIMBD,Yahya2025Tparam} & SIM (D/BD-RIS) & $S$/$T$-param., num. & circuit \\
\cite{Abrardo2025SIM,Pettanice2026SIMvalid,Abrardo2026nonlinear} & SIM (ECO) & $Z$/$S$-par., FW & circuit \\
\cite{Papazafeiropoulos2026SIMFAS} & SIM Tx + FAS Rx & none & statistical \\
\cite{Shen2025LIM,Zhao2025MERIS,Ye2025fRIS,Zhu2025FRISsec,Zhu2026FRISIM,Liu2025FSTARNOMA,MESTARS2024} & liquid, ME, fluid (STAR-)RIS & none & spatial \\
\cite{Xu2026MABDRIS} & movable BD-RIS & none & spatial \\
\cite{Li2023BDRISMC,Nerini2024closedform,Qian2020MC,Sheemar2026STARBD} & (STAR) BD-RIS, RIS & $S$/$Z$-param. & circuit \\
\cite{Pizzo2025MC,Wang2025CAPAMC} & HMIMO / CAPA & kernel, closed form & wavenumber \\
\cite{Zhu2026MAMC,Xu2026MAdir,Liao2026MAMC} & MA array & $Z$-param., elementary radiators & circuit \\
\cite{Iacovelli2026MC} & FAS/MA on surface & full kernel $\kerZ$ & wavenumber \\
This work & fluid Tx + fluid BD stack & 3-D kernel $\kerZ$, cl.\ form & wavenumber \\
\bottomrule
\end{tabular}
\end{table}

\subsection{Contributions}

Our contributions are the following.
\begin{itemize}
\item \textbf{Three-dimensional circuit-field equivalence.} The impedance matrix of any constellation of ports on parallel planes is the kernel sampled at the three-dimensional separations, in the closed spherical-Hankel form of \cite{Iacovelli2026MC} (Theorem~\ref{thm:cfe3d}): intra-layer coupling, inter-layer transfer, and feedback onto the transmitter are one function. The Rayleigh--Sommerfeld coefficient of the SIM literature is its scalar, unilateral counterpart with a flat spectral weight (Remark~\ref{rem:RS}).
\item \textbf{Induced currents, and the cascade as a limit.} Terminating the stack on one reconfigurable network gives the induced currents, the multi-user channel, and the stack-loaded input impedance in one inversion (Section~\ref{sec:loaded}). The SIM cascade is its single-pass limit under two explicit assumptions, whose failure in compact stacks is structural (Proposition~\ref{prop:cascade}, Remark~\ref{rem:stiffness}), and the transmitter pays for the dissipation of every sheet (Lemma~\ref{lem:passivity}).
\item \textbf{Interconnection across sheets.} The network may join atoms on different sheets (Remark~\ref{rem:interlayer}). Vertical networks repeated at every site are spectrally local and terminate each modal transmission line with a lossless reciprocal $L$-port (Corollary~\ref{cor:vertical}).
\item \textbf{Wavenumber dichotomy with propagator.} Across a separation $h$ the kernel spectrum acquires $e^{\im k_zh}$ (Theorem~\ref{thm:dichotomy3d}); the reaction splits in closed form into a Fresnel and an evanescent part, the latter dominating below $h\approx0.17\lambda$ (Corollary~\ref{cor:axial}); Fourier modes diagonalize every inter-plane block (Lemma~\ref{lem:diag3d}), each mode sees a loaded transmission line (Corollary~\ref{cor:tline}), and a fluid port on any plane is a constant-modulus codeword (Proposition~\ref{prop:codeword3d}).
\item \textbf{Coupling-aware joint design.} The weighted sum spectral efficiency is maximized over precoders, positions on every plane, and loads under the stack-loaded power, voltage, and induced-current constraints, with closed-form gradients through one inversion and a holographic bound as certificate (Section~\ref{sec:design}), and evaluated with every coupling entry computed for the actual element profile, which shows, among other things, that fluid ports on a $\lambda/8$ grid perform within a few per cent of continuously movable ones (Section~\ref{sec:numerics}).
\end{itemize}

\subsection{Notation}
Boldface lowercase (uppercase) letters denote vectors (matrices), calligraphic letters sets, surfaces, and kernels, and $(\cdot)^{\mathsf T}$, $(\cdot)^{\mathsf H}$, $(\cdot)^*$ transpose, conjugate transpose, and conjugate. $\Real\{\cdot\}$ and $\Imag\{\cdot\}$ act entrywise, $\im$ is the imaginary unit, and the convention $e^{-\im2\pi ft}$ makes outgoing waves carry $e^{+\im\kappa R}$, $\kappa=2\pi/\lambda$. Points are $\br=[\brb^{\mathsf T},z]^{\mathsf T}$ with planar part $\brb\in\mathbb R^2$, planar wavevectors are $\bk$ with $k_z=\sqrt{\kappa^2-\|\bk\|^2}$, $\Imag\{k_z\}\ge0$. $j_\ell,y_\ell,h_\ell=j_\ell+\im y_\ell$ are the spherical Bessel and Hankel functions, $\delta_{ii'}$ the Kronecker delta, $\bI_d$ the identity, $\be_n$ the $n$-th canonical vector. Indices: $l,l'\in\{0,\ldots,L\}$ label planes, plane $0$ being the transmitter; $n,m$ ports within a plane; $k,k'$ users; $i,i'$ Fourier modes. The block subscripts $0$ and $\mathrm s$ denote the transmit ports and the union of all stack ports, respectively; $N_{\rm s}=\sum_{l\ge1}N_l$ is the number of stack ports and $\bW=[\bw_1,\ldots,\bw_K]$ collects the precoders.

\section{Multi-Plane Electromagnetic Model}\label{sec:model}

\subsection{Sheets, kernel, and complex power}

Consider a quasi-static narrowband system with $L+1$ parallel planes $\plane_l=\{z=z_l\}$, $l=0,\ldots,L$, with $0=z_0<z_1<\cdots<z_L$. Plane $0$ hosts the transmit ports. Planes $1,\ldots,L$ host the sheets of the stack. The $K$ single-antenna users lie beyond the last sheet, $z>z_L$, so that every sheet is interposed between the transmitter and the users. Each plane carries a rectangular aperture $\surf_l\subset\plane_l$, modeled as a thin conducting sheet of surface resistance $Z_s\ge0$, and a tangential monochromatic current density $\bj_l(\bsb)\in\mathbb C^3$. The total current is $\bj(\br)=\sum_l\bj_l(\brb)\,\delta(z-z_l)$ and it radiates, in free space,
\begin{equation}
\fieldE(\br)=\im\kappa Z_0\sum_{l=0}^{L}\int_{\surf_l}\bG(\br,\mathbf s)\bj_l(\bsb)\,d\bsb,\quad \mathbf s=[\bsb^{\mathsf T},z_l]^{\mathsf T},
\end{equation}
with $Z_0$ the impedance of vacuum and $\bG$ the dyadic Green function, depending on its arguments only through $\bm\tau=\br-\mathbf s$, with closed forms \cite[Lemma~1]{Iacovelli2026MC}
\begin{align}
\Imag\{\bG(\bm\tau)\}&=\tfrac{\kappa}{4\pi}\big[(j_0-\tfrac{j_1}{x})\bI_3+j_2\hat{\bm\tau}\hat{\bm\tau}^{\mathsf T}\big],\label{eq:ImG}\\
\Real\{\bG(\bm\tau)\}&=-\tfrac{\kappa}{4\pi}\big[(y_0-\tfrac{y_1}{x})\bI_3+y_2\hat{\bm\tau}\hat{\bm\tau}^{\mathsf T}\big],\label{eq:ReG}
\end{align}
$x=\kappa\|\bm\tau\|$, $\hat{\bm\tau}=\bm\tau/\|\bm\tau\|$. Nothing in (\ref{eq:ImG})--(\ref{eq:ReG}) is planar: $\bm\tau$ is a three-dimensional separation, and this is the observation on which the paper rests.

The complex power that must be delivered to sustain $\bj$ is, exactly as in \cite[Sec.~II-B]{Iacovelli2026MC}, the quadratic form
\begin{equation}
P_{\rm cx}=\tfrac12\langle\bj,\kerZ\bj\rangle,\label{eq:Pcx}
\end{equation}
of the impedance kernel
\begin{equation}
\kerZ(\br,\br')=Z_s\delta(\brb-\brb')\ind{z=z'}\bI_3-\im\kappa Z_0\bG(\br,\br'),\label{eq:kernel}
\end{equation}
which splits as $\kerZ=\kerR+\im\kerX$ with $\kerR=Z_s\delta\,\ind{z=z'}\bI_3+\kappa Z_0\Imag\{\bG\}$ and $\kerX=-\kappa Z_0\Real\{\bG\}$, both real symmetric. At short range $\kerR$ stays bounded, since $j_0\to1$ and $j_1/x\to1/3$, whereas $y_2\sim-3x^{-3}$ makes $\kerX$ grow as $R^{-3}$: the singularity is reactive alone. The power balance of \cite[Prop.~1]{Iacovelli2026MC} holds for the union of sheets, its proof integrating Poynting's theorem over a ball containing all of them:
\begin{equation}
\Real\{P_{\rm cx}\}=\sum_{l=0}^{L}P_{\Omega,l}+P_{\rm rad},\qquad \Imag\{P_{\rm cx}\}=2\omega(W_e-W_m),\label{eq:balance}
\end{equation}
with $P_{\Omega,l}$ the Joule dissipation on sheet $l$, $P_{\rm rad}$ the radiated power, and $W_e,W_m$ the stored energies. The ohmic delta acts on \emph{every} sheet, so an induced current dissipates in its own conductor exactly as an impressed one does, whereas the radiative and reactive kernels connect all sheets to each other.

\subsection{Ports on planes and the 3-D circuit-field equivalence}

Plane $l$ carries $N_l$ identical, translated elements. Element $n$ is a port of profile $\xi(\bsb-\bu_{l,n})\bp$, built on the common real, even profile $\xi$ of support diameter $\Delta\ll\lambda$, $\int\xi=\ell_{\rm e}$, and the common tangential polarization $\bp\perp\hat{\mathbf z}$. The port has the three-dimensional position $\br_{l,n}=[\bu_{l,n}^{\mathsf T},z_l]^{\mathsf T}$, and the collection of all planar positions is $\Utt=\{\bu_{l,n}\}$. Any pair of ports is separated by a planar offset and an axial offset, and both are needed throughout: we write
\begin{equation}
\btau\triangleq\bu_{l,n}-\bu_{l',m},\quad h\triangleq z_l-z_{l'},\quad R\triangleq\sqrt{\|\btau\|^2+h^2},\label{eq:seps}
\end{equation}
so that $h=0$ inside a plane and $R$ is the distance the kernel sees. With port currents $\biota_l\in\mathbb C^{N_l}$ (impressed by the generators on plane $0$, induced through the loads on planes $l\ge1$) the sheet current is $\bj_l=\sum_n\iota_{l,n}\xi(\cdot-\bu_{l,n})\bp$, and we stack $\biota=[\biota_0^{\mathsf T},\ldots,\biota_L^{\mathsf T}]^{\mathsf T}\in\mathbb C^{N}$, $N=\sum_l N_l$.

\begin{theorem}[Three-dimensional circuit-field equivalence]\label{thm:cfe3d}
\mbox{}

(i) Substituting the port currents into (\ref{eq:Pcx}) gives its port equivalent $P_{\rm cx}=\tfrac12\biota^{\mathsf H}\bZ(\Utt)\biota$ with the complex symmetric multiport impedance matrix $\bZ(\Utt)\in\mathbb C^{N\times N}$, whose entries are
\begin{equation}
[\bZ(\Utt)]_{(l,n),(l',m)}=\big(\xi\star\zeta_p(\cdot\,;h)\star\xi\big)(\btau),\label{eq:Zports}
\end{equation}
with $\star$ the planar correlation and the kernel restricted to the polarization
\begin{equation}
\zeta_p(\btau;h)\triangleq Z_s\delta(\btau)\ind{h=0}-\im\kappa Z_0\bp^{\mathsf T}\bG([\btau^{\mathsf T},h]^{\mathsf T})\bp.\label{eq:zetap}
\end{equation} It is the unique matrix matching the complex power of every excitation, and $\Real\{\bZ\}$, $\Imag\{\bZ\}$ are the port-sampled resistive and reactive kernels $\kerR$ and $\kerX$.

(ii) As $\kappa\Delta\to0$, every mutual entry, on the same plane or across planes, converges to the kernel sampled at the three-dimensional separation,
\begin{equation}
[\bZ]_{(l,n),(l',m)}\to R_r\,\vartheta(\kappa R,c),\qquad c\triangleq\bpb^{\mathsf T}\btau/R,\label{eq:theta3d}
\end{equation}
$R_r=Z_0\kappa^2\ell_{\rm e}^2/(6\pi)$, and
\begin{equation}
\vartheta(x,c)=\rho+\im\chi=\tfrac32\Big[h_0(x)-\tfrac{h_1(x)}{x}+h_2(x)\,c^2\Big],\label{eq:theta}
\end{equation}
the function of \cite[Eq.~(17)]{Iacovelli2026MC}, with $\rho$ ($j_\ell$ family) the normalized resistive kernel and $\chi$ ($y_\ell$ family) the normalized reactive kernel. The self-terms are as in \cite[Thm.~1(iii)-(iv)]{Iacovelli2026MC} and are the same on every plane: the radiation self-resistance tends to $R_r$, the ohmic self-resistance is $R_\Omega=\varepsilon R_r$, and the self-reactance is $X_A=\zeta_AR_r$.

(iii) $\bZ$ is symmetric, $[\bZ]_{(l,n),(l',m)}=[\bZ]_{(l',m),(l,n)}$, so every inter-plane block satisfies $\bZ_{ll'}=\bZ_{l'l}^{\mathsf T}$: the transfer from plane $l'$ to plane $l$ and the backward transfer from $l$ to $l'$ are the same numbers, the reflection of a sheet being its diagonal block $\bZ_{ll}$ together with its load.
\end{theorem}
\begin{proof}
See \apx{app:cfe3d}{A}.
\end{proof}

Consequently,
\begin{align}
\bZ(\Utt)&=R_r\Big[(1+\varepsilon+\im\zeta_A)\bI_N+\bTheta(\Utt)\Big],\label{eq:Znorm}\\
[\bTheta]_{(l,n),(l',m)}&=\begin{cases}\vartheta(\kappa R,c), & (l,n)\ne(l',m),\\ 0, & (l,n)=(l',m),\end{cases}\nonumber
\end{align}
For a lone driven sheet, $L=0$, the real part of (\ref{eq:Znorm}) is $R_r\bC_\varepsilon$, with the single-plane matrix
\begin{equation}
\bC_\varepsilon\triangleq\varepsilon\bI_{N_0}+\bC,\qquad [\bC]_{nm}=\rho(\kappa R,c),\quad [\bC]_{nn}=1,\label{eq:Ceps}
\end{equation}
which prices radiation through $\bC$ and conduction loss through $\varepsilon\bI_{N_0}$ \cite[Eq.~(19)]{Iacovelli2026MC} and enters the power constraint of every coupling-aware single-layer design. What replaces it in front of a stack is derived next. The limit (\ref{eq:theta3d}) requires $\Delta\ll R$, which facing ports of adjacent sheets, only $h$ apart, may violate: for a Gaussian profile of standard deviation $\sigma$ the facing reactance of (\ref{eq:theta3d}) exceeds that of (\ref{eq:Zports}) by factors of $12$, $5$, and $1.2$ at $h=1.25\sigma$, $2\sigma$, and $5\sigma$. Compact stacks are therefore evaluated with the profile, $\vartheta$ being replaced in (\ref{eq:Znorm}) by the smoothed kernel of (\ref{eq:Zports}) and the unit radiation self-resistance by $\rho_\sigma(0)\le1$, the element spectrum $\hat\xi^2$ tapering the visible disc: the point limit is used where it holds, $\Delta\ll R$.

The diagonal blocks $\bTheta_{ll}$ are the in-plane coupling of \cite{Iacovelli2026MC} and the off-diagonal blocks $\bTheta_{ll'}$ the inter-plane transfer, with no new function for either. For facing ports, $\btau=\mathbf0$, one has $c=0$, so the axial transfer is the polarization-independent H-plane function $\tfrac32[h_0-h_1/x]$ at $x=\kappa h$, and for every other pair the axial offset keeps $|c|<1$, so the E-plane extreme of the in-plane kernel is never reached across planes. Moreover $R\ge h>0$, so no inter-plane entry is singular and the $R^{-3}$ growth of $\kerX$ is confined to the diagonal blocks.

\begin{remark}[The Rayleigh--Sommerfeld coefficient]\label{rem:RS}
The SIM literature models the transfer from atom $m$ of layer $l-1$ to atom $n$ of layer $l$ by the scalar first Rayleigh--Sommerfeld coefficient \cite{An2023SIM,Wei2026FSIM,Sheemar2026SIMsurvey} $w_{nm}=-2A_{\rm a}\partial_z e^{\im\kappa R}/(4\pi R)=A_{\rm a}(h/R^2)(1/(2\pi R)-\im/\lambda)e^{\im\kappa R}$, with $A_{\rm a}$ the atom area, $\theta$ the angle of the ray from $\hat{\mathbf z}$ and $\varphi$ its azimuth from $\bpb$, so that $h/R=\cos\theta$ and $c=\sin\theta\cos\phi$. In the far zone, $\kappa R\gg1$, (\ref{eq:theta}) gives $\vartheta\approx\tfrac32(1-c^2)e^{\im\kappa R}/(\im\kappa R)$ against $w_{nm}\approx-\im A_{\rm a}\cos\theta\,e^{\im\kappa R}/(\lambda R)$, the dipole pattern $1-c^2$ in place of the obliquity $\cos\theta$: the two agree at broadside, but the kernel falls as $\cos^2\theta$ in the E-plane ($\varphi=0$) and not at all in the H-plane ($\varphi=\pi/2$), so the polarization-blind coefficient overestimates the E-plane transfer by $1/\cos\theta$ and underestimates the H-plane one by $\cos\theta$. In the near zone, where compact stacks operate, (\ref{eq:theta}) is led by the static dipole term $\propto(3c^2-1)/x^3$, which changes sign on the cone $c^2=1/3$, at the lateral offset $h/\sqrt2$ in the E-plane, whereas $w_{nm}\approx A_{\rm a}h/(2\pi R^3)$ keeps its sign; and $w_{nm}$ acts from layer $l-1$ to $l$ alone, whereas (\ref{eq:Zports}) connects every pair of planes. The models agree up to a constant only near broadside and without intra-layer coupling or inter-layer reflection \cite{Nerini2024SIMBD}, not the regime of compact stacks \cite{Wei2026FSIM}.
\end{remark}

\subsection{Loads, interconnection patterns, and induced currents}\label{sec:loaded}

Only plane $0$ is driven. Its currents $\biota_0\in\mathbb C^{N_0}$ are impressed by the generators, one RF chain per port, and are the design variable. Every port of every stack sheet is instead terminated on a single reconfigurable network, of impedance matrix
\begin{equation}
\bM=\bR_M+\im\bX\in\mathbb C^{N_{\rm s}\times N_{\rm s}},\qquad \bX\in\netset,\quad \bR_M\succeq0,\label{eq:load}
\end{equation}
with $\bX$ and $\bR_M$ real symmetric by reciprocity. The default $\bR_M=\mathbf 0$ is the lossless network of the BD-RIS literature \cite{Shen2022BDRIS,Nerini2024graph}; a nonzero $\bR_M$ models the RF loss of the network itself, set by its technology, distinct from the element loss $\varepsilon$ and the binding constraint on long connections \cite{Nerini2025lossy,Peng2026lossy}.

The interconnection pattern $\netset\subseteq\mathbb S^{N_{\rm s}}$ is a linear subspace of real symmetric matrices, the support of a reactance graph on \emph{all} stack ports, deliberately not required to be block-diagonal by layer. With $\bX=[\bX_{ll'}]$ in per-plane blocks, the families of interest are
\begin{itemize}
\item \emph{layer-diagonal} patterns, $\bX_{ll'}=\mathbf0$ for $l\ne l'$, inside which each layer carries its own conventional architecture: diagonal for a D-RIS layer, block-diagonal for group-connected, tree or arrowhead for tree-connected, dense for fully connected \cite{Nerini2024graph,Li2025BDRIStutorial}. This is the family the stacked-metasurface literature assumes.
\item \emph{vertically connected} patterns, in which every off-diagonal block $\bX_{ll'}$ is diagonal, so that connections join the atoms of one in-plane site across sheets and each site owns a lossless $L$-port down the depth of the stack. That $L$-port may join all $L(L-1)/2$ pairs of sheets at the site, or only adjacent ones, the $L-1$ blocks $\bX_{l,l+1}$ of a nearest-neighbour ladder.
\item \emph{stack-wide} patterns, up to the full $\mathbb S^{N_{\rm s}}$, in which any two atoms of any two sheets may be joined.
\end{itemize}
The families are nested, so their optima are ordered at fixed geometry, and nothing below uses block-diagonality: every result needs only $\bM$ complex symmetric with $\Real\{\bM\}\succeq0$ and $\bX$ in a linear subspace.

\begin{remark}[Inter-layer connections are the shorter ones]\label{rem:interlayer}
What limits a reactance network is the electrical length of its connections, which must stay short for the lumped model to hold \cite{Nerini2024graph,Li2025BDRIStutorial}. Inter-layer connections span only the inter-sheet distance, $h\in[\lambda/16,\lambda/2]$ in a compact stack, shorter than most connections within one panel. The canonical transmissive unit cell, two facing radiators joined by a tunable two-port \cite{Li2022BDRIS,Nerini2024SIMBD,Sheemar2026SIMsurvey}, is itself one. Layer-diagonal modeling is therefore an over-restrictive assumption.
\end{remark}

Two distinct relations hold. The first is the \emph{multiport relation} $\bv=\bZ(\Utt)\biota$, which in block form reads
\begin{equation}
\begin{bmatrix}\bv_0\\ \bv_{\rm s}\end{bmatrix}
=\begin{bmatrix}\bZ_{00}&\bZ_{0\rm s}\\ \bZ_{\rm s0}&\bZ_{\rm ss}\end{bmatrix}
\begin{bmatrix}\biota_0\\ \biota_{\rm s}\end{bmatrix},\label{eq:multiport}
\end{equation}
where $\bZ_{\rm ss}$ contains the intra-plane blocks $\bZ_{ll}$ and the inter-plane blocks $\bZ_{ll'}$, $l,l'\ge1$. The second is the \emph{terminal law} of a loaded multiport \cite{Li2023BDRISMC,Nerini2024closedform}: the stack ports see no generator, only the network, so that, with $\biota_{\rm s}$ leaving the network,
\begin{equation}
\bv_{\rm s}=-\bM\,\biota_{\rm s}.\label{eq:terminal}
\end{equation}
Substituting (\ref{eq:terminal}) into the second row of (\ref{eq:multiport}) and solving,
\begin{equation}
(\bZ_{\rm ss}+\bM)\,\biota_{\rm s}=-\bZ_{\rm s0}\biota_0
\ \Longrightarrow\
\biota_{\rm s}=-\bT\bZ_{\rm s0}\biota_0\triangleq-\bB\,\biota_0,\label{eq:induced}
\end{equation}
with $\bT\triangleq(\bZ_{\rm ss}+\bM)^{-1}$, and the first row of (\ref{eq:multiport}) then gives $\bv_0=\bZ_{\rm in}\biota_0$ with the Schur complement (\ref{eq:Zin}).

\begin{lemma}[Well-posedness, input impedance, and where the power goes]\label{lem:passivity}
For $\varepsilon>0$ the matrix $\bZ_{\rm ss}+\bM$ is invertible for every $\bX$ and every geometry, since its Hermitian part is $\Real\{\bZ_{\rm ss}\}+\bR_M\succeq\varepsilon R_r\bI_{N_{\rm s}}$. The transmitter sees the stack-loaded input impedance
\begin{equation}
\bZ_{\rm in}(\Utt,\bM)=\bZ_{00}-\bZ_{0\rm s}\bT\bZ_{\rm s0},\label{eq:Zin}
\end{equation}
complex symmetric, whose real part prices the total average power drawn from the generators,
\begin{equation}
\tfrac12\biota_0^{\mathsf H}\Real\{\bZ_{\rm in}\}\biota_0=\underbrace{P_{\Omega,0}}_{\text{Tx sheet}}+\underbrace{\sum_{l=1}^{L}P_{\Omega,l}}_{\text{stack sheets}}+\underbrace{P_{M}}_{\text{network}}+P_{\rm rad},\label{eq:power3d}
\end{equation}
with $P_{\Omega,l}=\tfrac{\varepsilon R_r}{2}\|\biota_l\|^2$ and $P_M=\tfrac12\biota_{\rm s}^{\mathsf H}\bR_M\biota_{\rm s}\ge0$. All terms being non-negative and $P_{\Omega,0}$ being present for every $\biota_0$, $\bC_{\rm in}\triangleq\Real\{\bZ_{\rm in}\}/R_r\succeq\varepsilon\bI_{N_0}\succ0$.
\end{lemma}
\begin{proof}
See \apx{app:passivity}{B}.
\end{proof}

Two readings of (\ref{eq:power3d}) matter. The induced currents dissipate on their own conductors, and this loss is billed to the transmitter through the Schur complement: a lossy stack is expensive not because its loads absorb, but because the currents it must sustain heat its sheets. And $\bC_{\rm in}$ replaces the single-plane $\bC_\varepsilon$ of (\ref{eq:Ceps}) by a position- and load-dependent matrix: the stack pulls the input impedance of the feed, an effect absent from the cascade model, represented numerically in \cite{Nerini2024SIMBD,Yahya2025Tparam} and given here in closed form.

\subsection{Effective channel and the cascade limit}\label{sec:effch}

User $k$ is a dipole of unit polarization $\bm\psi_k$ at $\br_k$ beyond the stack, observing $y_k=\int\bg_k^{\mathsf H}\bj+n_k$ with $\bg_k^{\mathsf H}(\mathbf s)=\im\kappa Z_0\bm\psi_k^{\mathsf H}\bG(\br_k,\mathbf s)$ and $n_k\sim\mathcal{CN}(0,\sigma_k^2)$. By linearity, $y_k=\bg_{k,0}^{\mathsf T}\biota_0+\bg_{k,\rm s}^{\mathsf T}\biota_{\rm s}+n_k$, where $[\bg_{k,l}]_n\triangleq\int\bg_k^{\mathsf H}(\mathbf s)\bp\,\xi(\bsb-\bu_{l,n})\,d\bsb$ is the profile-smoothed channel at the port position, the 3-D version of \cite[Eq.~(33)]{Iacovelli2026MC}. Substituting (\ref{eq:induced}),
\begin{equation}
y_k=\tilde\bg_k^{\mathsf T}\biota_0+n_k,\qquad \tilde\bg_k^{\mathsf T}(\Utt,\bM)\triangleq\bg_{k,0}^{\mathsf T}-\bg_{k,\rm s}^{\mathsf T}\bT\bZ_{\rm s0}.\label{eq:effchannel}
\end{equation}
The channel (\ref{eq:effchannel}) is the impedance-form counterpart of the loaded-scatterer one of \cite{Li2023BDRISMC,Abrardo2025SIM}, closed-form in the positions and with one $N_{\rm s}\times N_{\rm s}$ inversion in place of the nested cascade of \cite{Nerini2024SIMBD}, and it contains every multiple reflection, intra-layer coupling, and feedback, since $\bZ_{\rm ss}$ couples every pair of planes.

\begin{proposition}[The cascade is the single-pass limit, and its layer response is amplitude-phase locked]\label{prop:cascade}
Equation (\ref{eq:induced}) is the refractive picture solved self-consistently, every sheet re-radiating forward and backward at once. Then:

(i) (\emph{recursion}) the $l$-th block row of (\ref{eq:induced}) reads, in full,
\begin{equation}
(\bZ_{ll}+\bM_{ll})\biota_l+\!\!\sum_{l'\ge1,\,l'\ne l}\!\!(\bZ_{ll'}+\bM_{ll'})\biota_{l'}=-\bZ_{l0}\biota_0,\label{eq:blockrow}
\end{equation}
the cross terms carrying the field coupling through $\bZ_{ll'}$ and the circuit coupling through $\bM_{ll'}$. Two assumptions reduce it to a one-step recursion. \emph{(A1)} the network is layer-diagonal, $\bM_{ll'}=\mathbf0$, so that no sheet is wired to another. \emph{(A2)} the field coupling is \emph{nearest-neighbour forward}: $\bZ_{ll'}$ is kept only for $l'=l-1$, the transmit plane included, so plane $0$ illuminates sheet $1$ alone, the direct terms $-\bZ_{l0}\biota_0$, $l\ge2$, are discarded, and no sheet couples backward; dropping only the backward blocks, $l'>l$, is the weaker \emph{(A2$^-$)}. Under (A1) and (A2), (\ref{eq:blockrow}) becomes
\begin{equation}
\biota_l=\bPsi_l\bP_l\,\biota_{l-1},\quad \bPsi_l\triangleq-(\bZ_{ll}+\bM_{ll})^{-1},\ \ \bP_l\triangleq\bZ_{l,l-1},\label{eq:onestep}
\end{equation}
for $l=1,\ldots,L$, each sheet depending only on the one before.

(ii) (\emph{cascade}) iterating (\ref{eq:onestep}) gives
\begin{equation}
\biota_L=\bPsi_L\bP_L\cdots\bPsi_1\bP_1\,\biota_0=\prod_{l=L}^{1}\big[\bPsi_l\bP_l\big]\,\biota_0,\label{eq:cascade}
\end{equation}
which is the cascade of per-layer responses and propagation matrices of \cite{An2023SIM,Wei2026FSIM}, with $\bP_l$ now the closed-form polarized kernel of Theorem~\ref{thm:cfe3d} and $\bPsi_l$ the physical layer response.

(iii) (\emph{locked response}) if in addition the intra-layer coupling is neglected and the load is diagonal and lossless, $[\bM_{ll}]_{nn}=\im X_{l,n}$, then
\begin{equation}
[\bPsi_l]_{nn}=\frac{-1}{R_r(1+\varepsilon)+\im(X_A+X_{l,n})},\label{eq:lorentz}
\end{equation}
and as the single real knob $X_{l,n}$ sweeps $\mathbb R$ this traces a circle through the origin of diameter $1/(R_r(1+\varepsilon))$. Amplitude and phase are therefore not independently assignable. The modulus is largest at the element resonance $X_{l,n}=-X_A$, where the response is the negative real $-1/(R_r(1+\varepsilon))$, and the phase departs from $\pi$ by less than $\pi/2$, approaching $\pi\pm\pi/2$ only as the modulus vanishes. The unit-modulus phase mask $\diag(e^{\im\theta_{ln}})$ of the cascade model is not realizable by one lossless reactance per element, and lifting the lock needs a second degree of freedom per atom.
\end{proposition}
\begin{proof}
(i) and (ii) are the statements above. (iii) is the M\"obius map $x\mapsto-1/(a+\im(b+x))$ with $a=R_r(1+\varepsilon)>0$, which sends the real line to a circle through the origin, and $|[\bPsi_l]_{nn}|=(a^2+(X_A+X_{l,n})^2)^{-1/2}$, $\arg[\bPsi_l]_{nn}=\pi-\arctan((X_A+X_{l,n})/a)\in(\pi/2,3\pi/2)$, whence the stated extremes.
\end{proof}

The terms discarded by (A2) are not small in a compact stack: the backward blocks $\bZ_{l,l+1}=\bZ_{l+1,l}^{\mathsf T}$ equal the forward ones (Theorem~\ref{thm:cfe3d}(iii)) and carry the near field of the gap. And (A1) excludes the vertically connected stacks of Remark~\ref{rem:interlayer}, which admit no cascade form however the field coupling is truncated.

\begin{remark}[Why the truncation fails in the direction it does]\label{rem:stiffness}
Consider two facing atoms tuned to resonance, $X_{l,n}=-X_A$. Their block of $\bZ_{\rm ss}+\bM$ is $\left[\begin{smallmatrix}a&z\\ z&a\end{smallmatrix}\right]$, with the small real $a=R_r(1+\varepsilon)$ and the mutual impedance $z$ in both triangles by Theorem~\ref{thm:cfe3d}(iii), and $|z|$ grows as the gap closes, the reactive kernel growing as $R^{-3}$. Once $|z|\gg a$ the inverse has norm close to $1/|z|$: the coupling detunes the pair, and a compact resonant stack partly chokes itself. (A2) keeps $\left[\begin{smallmatrix}a&0\\ z&a\end{smallmatrix}\right]$, whose inverse has the entry $-z/a^2$, so the cascade predicts the largest currents where the true stack carries the smallest. (A2) also drops the direct blocks $\bZ_{l0}$, $l\ge2$, which are set by the standoff $z_1$ rather than by $h$ and do not fade as the stack is stretched; (A2$^-$) keeps them, and for $L=2$ it differs from (A2) by exactly these blocks.
\end{remark}

\begin{corollary}[Efficiency identity and the control-loss tension]\label{cor:eff}
The radiation efficiency of the whole structure, radiated power over power drawn from the transmit generators, is
\begin{equation}
\eta=\frac{P_{\rm rad}}{P_{\rm rad}+\tfrac{\varepsilon R_r}{2}\big(\|\biota_0\|^2+\sum_{l=1}^{L}\|\biota_l\|^2\big)+P_M},\label{eq:eta}
\end{equation}
so that, at fixed radiated power, the price of the stack is set by the total induced-current norm $\sum_{l\ge1}\|\biota_l\|^2$ and not by the layer count. Layer $l$ acts on the field only through its current, its contribution to $\tilde\bg_k^{\mathsf T}\biota_0$ being $\bg_{k,l}^{\mathsf T}\biota_l$, so a layer that costs nothing does nothing. The norm grows geometrically, as a compact stack couples the transmitter to every sheet rather than letting the first shadow the rest, and electrically, as an atom carries the most current when its load resonates the element, $X_{l,n}\approx-X_A$, exactly where the phase of (\ref{eq:lorentz}) moves fastest: strong control and strong dissipation are the same operating point.
\end{corollary}

\subsection{What a single sheet can do}\label{sec:sheet}

\begin{lemma}[Sheet scattering identity]\label{lem:sheet}
Let a plane wave of transverse wavenumber $\bk$, $\|\bk\|<\kappa$, and unit transverse amplitude be incident from $z<z_l$ on a sheet at $z=z_l$ carrying tangential currents, and let $t(\bk)$ and $r(\bk)$ be the normalized transverse amplitudes of the total field above the sheet and of the field returning below it. A tangential current sheet radiates the same transverse amplitude $\beta$ into both half-spaces, so $t=1+\beta$ and $r=\beta$, i.e., mode by mode,
\begin{equation}
t(\bk)=1+r(\bk).\label{eq:star}
\end{equation}
If the sheet is lossless, $|t|^2+|r|^2=1$ forces $r=\tfrac12(e^{\im\psi}-1)$ and $t=\tfrac12(e^{\im\psi}+1)$ for some $\psi$. Factoring the half angle,
\begin{equation}
t=\tfrac12e^{\im\psi/2}\big(e^{\im\psi/2}+e^{-\im\psi/2}\big)=\cos(\psi/2)\,e^{\im\psi/2},\label{eq:halfangle}
\end{equation}
so that the power split $|t|^2=\cos^2(\psi/2)$ and the transmission phase $\arg t=\psi/2$ are two readings of the same parameter and cannot be set independently. At the operator level, on the $N_{\Iset}$ modes of Section~\ref{sec:modes3d}, $\bT_{\rm sh}=\tfrac12(\bU+\bI_{N_{\Iset}})$ with $\bU=\bQ\diag(e^{\im\psi_i})\bQ^{\mathsf T}$ unitary and $\bQ$ real orthogonal, so the statement holds in every eigenchannel of the sheet, the Fourier modes for a shift-invariant one.
\end{lemma}
\begin{proof}
See \apx{app:star}{C}.
\end{proof}

Normalized by $R_r$, the circle traced by (\ref{eq:lorentz}) for a lossless element is exactly this locus of $r$, the Lorentz form being its Cayley parameterization, and conduction loss contracts it by $1/(1+\varepsilon)$. A single sheet cannot refract without reflecting: any $\beta=r\ne0$ that alters the transmitted field sends an equal wave backwards, and transmission phase is bought with transmitted power on a fixed curve, $|t|=|\cos(\arg t)|$ by (\ref{eq:halfangle}), so an eighth of a wave of phase, $\arg t=\pm\pi/4$, already costs $3$~dB, and at $\arg t=\pm\pi/2$ the whole incident power returns to the transmitter, where $\bZ_{\rm in}$ accounts for it. The unit-modulus phase mask of the cascade model thus describes no single sheet; a stack of $L$ such layers is physically $2L$ planes, and reflectionless operation with free phase, the Huygens condition, needs a two-sheet layer, i.e., the vertically connected pattern of Remark~\ref{rem:interlayer}.

\subsection{Multi-user signal model and constraints}\label{sec:constraints}

The multi-user transmit signal superposes $K$ precoded streams, $\biota_0=\sum_k\bw_kx_k$, and the SINR of user $k$ is
\begin{equation}
\Gamma_k=\frac{|\tilde\bg_k^{\mathsf T}\bw_k|^2}{\sum_{k'\ne k}|\tilde\bg_k^{\mathsf T}\bw_{k'}|^2+\sigma_k^2}.\label{eq:sinr}
\end{equation}
Three constraints accompany it, each a quadratic form of the precoders. Absorbing $R_r/2$ into $P$, the \emph{power} constraint $\sum_k\bw_k^{\mathsf H}\bC_{\rm in}\bw_k\le P$ bills the transmitter for radiation and for the dissipation of every sheet and load (Lemma~\ref{lem:passivity}). The \emph{transmit voltage} constraint $\sum_k|\be_n^{\mathsf T}\bar\bZ_{\rm in}\bw_k|^2\le\bar V_{\max}^2$, $\bar\bZ_{\rm in}=\bZ_{\rm in}/R_r$, caps the amplifier voltage, which coupled arrays can make large at modest delivered power \cite[Sec.~III-C]{Iacovelli2026MC}. The \emph{induced-current} constraint $\sum_k|\be_n^{\mathsf T}\bB\bw_k|^2\le\bar I_{\max}^2$, specific to a loaded stack, protects the tunable components that carry $\biota_{\rm s}=-\bB\biota_0$ and keeps the design away from the regime where Corollary~\ref{cor:eff} bites hardest.

\section{Wavenumber-Domain Representation Across Planes}\label{sec:wavenumber}

In the wavenumber domain the light-circle dichotomy of \cite[Prop.~2]{Iacovelli2026MC} survives across planes with a propagator attached, Fourier modes diagonalize every block of $\bZ$, and a fluid port on any plane remains a codeword.

\subsection{The light circle with a propagator}

\begin{theorem}[Light-circle dichotomy with propagator]\label{thm:dichotomy3d}
For $h\ge0$, in the sense of tempered distributions, $\zeta_p(\btau;h)=\frac{1}{4\pi^2}\int_{\mathbb R^2}\hat\zeta_p(\bk;h)e^{\im\bk\cdot\btau}d\bk$ with
\begin{equation}
\hat\zeta_p(\bk;h)=Z_s\ind{h=0}+\kappa Z_0\frac{1-c_{\bk}^2}{2k_z}\,e^{\im k_zh},\qquad c_{\bk}=\frac{\bk\cdot\bpb}{\kappa},\label{eq:zetahat3d}
\end{equation}
$k_z=\sqrt{\kappa^2-\|\bk\|^2}$, $\Imag\{k_z\}\ge0$. Consequently the resistive and reactive spectra across the separation $h$ are
\begin{equation}
\hat r_p(\bk;h)=Z_s\ind{h=0}+\kappa Z_0\frac{1-c_{\bk}^2}{2k_z}\cos(k_zh)\,\ind{\|\bk\|<\kappa},\label{eq:rhat3d}
\end{equation}
\begin{align}
\hat x_p(\bk;h)&=\kappa Z_0\frac{1-c_{\bk}^2}{2k_z}\sin(k_zh)\,\ind{\|\bk\|<\kappa}\nonumber\\
&\quad-\kappa Z_0\frac{1-c_{\bk}^2}{2|k_z|}e^{-|k_z|h}\,\ind{\|\bk\|>\kappa}.\label{eq:xhat3d}
\end{align}
At $h=0$ they reduce to \cite[Eqs.~(22)-(23)]{Iacovelli2026MC}. The weight $(1-c_{\bk}^2)/(2k_z)$ is the obliquity-and-polarization factor that the scalar coefficient of Remark~\ref{rem:RS} replaces by a flat weight, the spectrum of $w_{nm}$ being $A_{\rm a}e^{\im k_zh}$. Radiation still lives on the visible disc only. Reaction now has a propagating part, the Fresnel phase $\sin(k_zh)$ on the disc, and an evanescent part that decays as $e^{-|k_z|h}$ off the disc.
\end{theorem}
\begin{proof}
See \apx{app:dichotomy3d}{D}.
\end{proof}

\begin{corollary}[Axial split, closed form]\label{cor:axial}
For facing ports ($\btau=\mathbf0$, $x=\kappa h$) the $c^2$ term of (\ref{eq:theta}) is absent, and splitting (\ref{eq:xhat3d}) into its disc and ring parts splits the normalized reactive coupling $\chi(0,h)=\tfrac32\big[-\tfrac{\cos x}{x}+\tfrac{\cos x}{x^3}+\tfrac{\sin x}{x^2}\big]$ exactly into $\chi=\chi_{\rm pr}+\chi_{\rm ev}$, with
\begin{align}
\chi_{\rm pr}(h)&=\tfrac32\Big[\frac1{2x}-\frac1{x^3}-\frac{\cos x}{x}+\frac{\sin x}{x^2}+\frac{\cos x}{x^3}\Big],\nonumber\\
\chi_{\rm ev}(h)&=\tfrac32\Big[\frac1{x^3}-\frac1{2x}\Big],\label{eq:axialsplit}
\end{align}
while the resistive coupling $\rho(0,h)=\tfrac32\big[\tfrac{\sin x}{x}+\tfrac{\cos x}{x^2}-\tfrac{\sin x}{x^3}\big]$ comes from the disc alone. The evanescent part carries the $R^{-3}$ singularity, changes sign at $x=\sqrt2$, and decays only as $\tfrac{3}{4x}$, and evaluating (\ref{eq:axialsplit}) numerically it dominates for $h\lesssim0.17\lambda$, carrying $93\%$ of $\chi$ at $h=0.1\lambda$, $67\%$ at $0.15\lambda$, and $21\%$ at $0.2\lambda$; so the near-rim evanescent modes are not suppressed in the aggregate. For elements of finite size the ring integrand carries $\hat\xi^2$, which tempers the $h^{-3}$ growth: for a Gaussian profile of standard deviation $\lambda/20$ the facing reactance at $h=\lambda/16$ evaluates to $1.9R_r$ rather than $23R_r$.
\end{corollary}
\begin{proof}
See \apx{app:axial}{E}.
\end{proof}

Below about a sixth of a wavelength, the regime the fluid-element SIM prefers \cite{Wei2026FSIM}, the inter-layer reaction is the near-field term that no propagator-based model sees; above it, it is the Fresnel phase, which the scalar propagator carries with the wrong weight and without polarization.

\subsection{Fourier modes across planes}\label{sec:modes3d}

Assume a common aperture $\surf=[-D_x/2,D_x/2]\times[-D_y/2,D_y/2]$ on every plane (a smaller transmit aperture is embedded in it), and expand the scalar current of plane $l$ on the orthonormal modes $\phi_i(\bsb)=e^{\im\bk_i\cdot\bsb}/\sqrt{|\surf|}$, $\bk_i=2\pi[i_x/D_x,i_y/D_y]^{\mathsf T}$, so that $\bj_l(\bsb)=\sum_i q_{li}\,\phi_i(\bsb)\,\bp$ with $\bq_l=[q_{li}]_{i\in\Iset}$ the modal coefficients of that plane. Projecting the port expansion of $\bj_l$ on the modes gives, as in \cite[Eq.~(31)]{Iacovelli2026MC},
\begin{equation}
q_{li}=\int_{\surf}\phi_i^*(\bsb)\,\bp^{\mathsf T}\bj_l(\bsb)\,\dd\bsb\ \ \Longrightarrow\ \ \bq_l=\bXi\bPhi_l\,\biota_l,\label{eq:qmap}
\end{equation}
with $\bXi\triangleq\diag\{\hat\xi(\bk_i)\}$ the element taper, $\hat\xi(\bk)=\int\xi(\bsb)e^{-\im\bk\cdot\bsb}\dd\bsb$, and $\bPhi_l\triangleq[\bphi(\bu_{l,1}),\ldots,\bphi(\bu_{l,N_l})]$ the constant-modulus codewords $[\bphi(\bu)]_i=|\surf|^{-1/2}e^{-\im\bk_i\cdot\bu}$, the conjugate modes sampled at the ports. Computations retain the modes $\Iset=\{i:\|\bk_i\|\le k_{\max}\}$, $N_{\Iset}=|\Iset|$, which contain the $N_{\rm vis}=O(\kappa^2|\surf|)$ visible modes and, for $k_{\max}>\kappa$, a ring of evanescent ones; the reaction between planes is carried by that ring. The complex power becomes $P_{\rm cx}=\tfrac12\bq^{\mathsf H}\bY\bq$ with $\bq$ stacking the $q_{li}$ and the modal Gram blocks $\bY^{(ll')}\in\mathbb C^{N_{\Iset}\times N_{\Iset}}$, $[\bY^{(ll')}]_{ii'}=\langle\phi_i\bp,\kerZ_{ll'}\phi_{i'}\bp\rangle$, whose exact form is (\ref{eq:zetahat3d}) integrated against two Dirichlet kernels as in \cite[Eq.~(28)]{Iacovelli2026MC}.

\begin{lemma}[Asymptotic diagonalization across planes]\label{lem:diag3d}
Fix $\delta>0$ and $h=|z_l-z_{l'}|$, and let $\kappa D_x,\kappa D_y\to\infty$.

(i) For visible modes $\|\bk_i\|,\|\bk_{i'}\|\le\kappa-\delta$, $[\bY^{(ll')}]_{ii'}=\hat\zeta_p(\bk_i;h)\delta_{ii'}+\epsilon_{ii'}$ with $|\epsilon_{ii'}|\le\omega_h(D_{\min}^{-1/2})+O(D_{\min}^{-\alpha})$ for every $\alpha<1/2$, $D_{\min}\triangleq\kappa\min(D_x,D_y)$, $\omega_h$ the modulus of continuity of $\hat\zeta_p(\cdot\,;h)$ on the disc of radius $\kappa-\delta/2$. Numerical evaluation further shows that the leakage of a mode onto its eight lattice neighbours falls as $D_{\min}^{-1}$, faster than this bound.

(ii) For $h>0$ the reactive Gram converges for every mode, evanescent modes included, to $\hat x_p(\bk_i;h)\delta_{ii'}$, because the propagator $e^{-|k_z|h}$ makes $\hat\zeta_p(\cdot\,;h)$ integrable against the $1/k$ tails of the truncated exponentials.

(iii) For $h=0$ the statements of \cite[Lemma~2]{Iacovelli2026MC} apply, the reactive Gram of a truncated mode diverging logarithmically. The cell-averaged modal weights and the Dirichlet leakage governing the rate are those of \cite{Iacovelli2026VA}.
\end{lemma}
\begin{proof}
See \apx{app:diag3d}{F}.
\end{proof}

In the limit the Gram blocks are therefore diagonal in the mode index, and we write
\begin{equation}
\bY^{(ll')}\to\bLam^{(ll')}\triangleq\diag\{\hat\zeta_p(\bk_i;|z_l-z_{l'}|)\},\label{eq:Lam}
\end{equation}
The distance between planes thus regularizes the rim for the reactive bookkeeping, as the element profile does in-plane: inter-plane reaction can be tracked mode by mode, and only the in-plane self-terms must be carried at port level \cite{Iacovelli2026MC}.

\begin{corollary}[Per-mode transmission line]\label{cor:tline}
By (\ref{eq:Lam}) the modal Gram $[\bY^{(ll')}]_{ii'}$ is diagonal in $(i,i')$, so it reorganizes into $N_{\Iset}$ blocks over plane indices, one per retained mode. Writing a mode index as a subscript and a plane pair as a superscript, (\ref{eq:zetahat3d}) gives
\begin{align}
[\bLam_i]_{ll'}&\triangleq[\bLam^{(ll')}]_{ii}=Z_s\delta_{ll'}+[\bar\bLam_i]_{ll'},\nonumber\\
[\bar\bLam_i]_{ll'}&=\frac{Z_{c,i}}{2}\,e^{\im k_{z,i}|z_l-z_{l'}|},\label{eq:tline}
\end{align}
with $Z_{c,i}=\kappa Z_0(1-c_{\bk_i}^2)/k_{z,i}$. The matrix $\bar\bLam_i$ is the mutual-impedance matrix of $L+1$ shunt current sources on a transmission line of propagation constant $k_{z,i}$ and characteristic impedance $Z_{c,i}$, the TE impedance $Z_0/\cos\theta_i$ in the H-plane and the TM impedance $Z_0\cos\theta_i$ in the E-plane, $\cos\theta_i=k_{z,i}/\kappa$, $Z_{c,i}/2$ being Wheeler's scan impedance \cite{Wheeler1965}, and $Z_s$ is the conduction loss of each sheet. A loaded layer is a shunt impedance on that line, and no layer acts alone within a mode.
\end{corollary}

Corollary~\ref{cor:tline} is the transverse-resonance picture of layered media, delivered as the wavenumber-domain face of the port model. Whether the loaded stack decouples depends on the network. Call it \emph{spectrally local} if all its blocks are diagonal in the wavenumber domain, $\bPhi_l\bM_{ll'}\bPhi_{l'}^{\mathsf H}=\diag_i\{\hat M_{ll',i}\}$, i.e., circulant on a port lattice common to the sheets. Mode $i$ then sees the $L\times L$ load $\hat\bM_i\triangleq[\hat M_{ll',i}]_{ll'}/\hat\xi^2(\bk_i)$, and its stack coefficients solve the modal image of (\ref{eq:induced}), $(\bar\bLam_i^{\rm ss}+Z_s\bI_L+\hat\bM_i)\bq_{{\rm s},(i)}=-\bar\bLam_i^{\rm s0}q_{0i}$, the superscripts selecting the sheet rows and the transmit column of $\bar\bLam_i$. A D-RIS layer, $\bM_{ll}=\im\diag\{X_{l,n}\}_n$, generally is not: $[\bPhi_l\bM_{ll}\bPhi_l^{\mathsf H}]_{ii'}$ is the Fourier coefficient of its reactance profile at $\bk_i-\bk_{i'}$, a convolution that mixes the modes unless all elements carry the same reactance. The simplest architecture in the wavenumber domain is thus the shift-invariant beyond-diagonal layer, a reversal of the element-domain ladder \cite{Li2022BDRIS,Nerini2024graph}.

\begin{corollary}[Ladder form of a spectrally local stack]\label{cor:ladder}
Let the load be spectrally local, so that mode $i$ sees the $L\times L$ load $\hat\bM_i$. Then $\bar\bLam_i^{\rm ss}$ is semiseparable and its inverse is tridiagonal: it is the nodal admittance matrix of the sheets as shunt nodes on the modal line, each coupled only to its neighbours through the section between them, of electrical length $\mu_l=k_{z,i}(z_{l+1}-z_l)$, and it exists unless a gap spans a multiple of half the axial wavelength, $\sin\mu_l=0$, which never happens for an evanescent mode. Multiplying the modal system by $(\bar\bLam_i^{\rm ss})^{-1}$ leaves $\big(\bI_L+(\bar\bLam_i^{\rm ss})^{-1}(Z_s\bI_L+\hat\bM_i)\big)\bq_{{\rm s},(i)}=-e^{\im k_{z,i}(z_1-z_0)}q_{0i}\,\be_1$, in which the transmitter enters the line at the first sheet only, delayed over the standoff $z_1-z_0$. If moreover the sheets are not interconnected, $\hat\bM_i$ is diagonal, the system is tridiagonal, and each mode costs $O(L)$ operations.
\end{corollary}
\begin{proof}
See \apx{app:ladder}{G}.
\end{proof}

\begin{corollary}[Site-invariant vertical networks are per-mode $L$-ports]\label{cor:vertical}
Let all sheets share a lattice matched to the modes, $\bPhi_l=\bPhi$ with $\bPhi\bPhi^{\mathsf H}=\bI_{N_{\Iset}}/A_{\rm c}$, $A_{\rm c}$ the cell area, and let the vertical network of Section~\ref{sec:loaded} be the same lossless $L$-port at every site, $\bM=\im\bX_{\rm v}\otimes\bI_{N_1}$ with $\bX_{\rm v}\in\mathbb S^{L}$. Then $\bPhi\bM_{ll'}\bPhi^{\mathsf H}=\im[\bX_{\rm v}]_{ll'}\bI_{N_{\Iset}}/A_{\rm c}$, so the network is spectrally local and every mode is terminated by the same $L$-port, $\hat\bM_i=\im\bX_{\rm v}/(A_{\rm c}\hat\xi^2(\bk_i))$, rather than by $L$ independent shunt reactances. A nearest-neighbour ladder, $\bX_{\rm v}$ tridiagonal with $2L-1$ reactances, keeps the system of Corollary~\ref{cor:ladder} banded and each mode at $O(L)$; the full $L$-port, $\bX_{\rm v}$ dense with $L(L+1)/2$, costs a coupler per column and $O(L^3)$ per mode, still below the $O(N_l^2)$ of one fully connected layer.
\end{corollary}

The useful dimension of a layer is then its visible mode count rather than its port count, the low-rank mechanism of \cite{Iacovelli2026XLBD}, so on a dense fluid grid the loads can be optimized in the mode domain. For a network that is not spectrally local, instead, the modal image is full, and $\bT$ of (\ref{eq:induced}) is best formed in the port domain.

\subsection{Fluid ports on every plane}

The two codeword maps of \cite{Iacovelli2026MC} carry over plane by plane.

\begin{proposition}[Positioning on planes]\label{prop:codeword3d}
Let $\bPhi_l$ and $\bXi$ be as in Section~\ref{sec:modes3d}, and $\bLam^{(ll')}$ as in (\ref{eq:Lam}).

(i) The profile-smoothed channel on plane $l$ is the transposed codeword map, dual to (\ref{eq:qmap}),
\begin{equation}
\bg_{k,l}=\bPhi_l^{\mathsf T}\bXi\,\hat\bg_k^{(l)},\label{eq:codewordmaps}
\end{equation}
where $\hat\bg_k^{(l)}$ collects the Fourier coefficients of the user functional on that plane,
\begin{equation}
\hat g^{(l)}_{k,i}=\frac{1}{\sqrt{|\surf|}}\int_{\surf_l}\bg_k^{\mathsf H}([\bsb^{\mathsf T},z_l]^{\mathsf T})\,\bp\;e^{\im\bk_i\cdot\bsb}\,\dd\bsb,\label{eq:ghat}
\end{equation}
so that $[\bg_{k,l}]_n=\sum_i\hat g_{k,i}^{(l)}\hat\xi(\bk_i)[\bphi(\bu_{l,n})]_i$ is a superposition of codeword evaluations.

(ii) In the diagonal limit of Lemma~\ref{lem:diag3d} the complex power splits over modes,
\begin{equation}
P_{\rm cx}=\tfrac12\sum_{i}\bq_{(i)}^{\mathsf H}\bLam_i\bq_{(i)},\label{eq:permodepower}
\end{equation}
with $\bq_{(i)}\triangleq[q_{0i},\ldots,q_{Li}]^{\mathsf T}$ the mode-major reordering of $\bq$ and $\bLam_i$ the per-mode block (\ref{eq:tline}); substituting (\ref{eq:qmap}) and collecting the port indices returns $\tfrac12\biota^{\mathsf H}\bZ\biota$, with $\bZ_{ll'}=\bPhi_l^{\mathsf H}\bXi\bLam^{(ll')}\bXi\bPhi_{l'}$, that is, with $\btau$ and $h$ as in (\ref{eq:seps}), the lattice sum
\begin{equation}
[\bZ_{ll'}]_{nm}=\frac{1}{|\surf|}\sum_i\hat\xi^2(\bk_i)\,\hat\zeta_p(\bk_i;|h|)\,e^{\im\bk_i\cdot\btau},\label{eq:pullback3d}
\end{equation}
every block being the pullback of the propagated kernel, $[\bZ_{ll'}]_{nm}\to(\xi\star\zeta_p(\cdot\,;|h|)\star\xi)(\btau)$: on cells of area $4\pi^2/|\surf|$, (\ref{eq:pullback3d}) is a Riemann sum of the integral $(4\pi^2)^{-1}\!\int\hat\xi^2\hat\zeta_p\,e^{\im\bk\cdot\btau}\dd\bk$ of Theorem~\ref{thm:dichotomy3d}, which closes the loop on (\ref{eq:Zports}).
\end{proposition}
\begin{proof}
(i) and (ii) follow from \cite[Thm.~2]{Iacovelli2026MC} applied plane by plane, with Lemma~\ref{lem:diag3d} in place of \cite[Lemma~2]{Iacovelli2026MC}. The convergence of the lattice sum for $h>0$ holds for the reactive part as well, by Lemma~\ref{lem:diag3d}(ii).
\end{proof}

The dictionary of the stack is thus complete: positioning is modulation, layering is propagation, and interconnection is a load, $\bPhi_l\bM_{ll'}\bPhi_{l'}^{\mathsf H}$ in the mode basis; movable and fluid layers are the same object over two position sets \cite[Remark~4]{Iacovelli2026MC}.

\section{Coupling-Aware Joint Design}\label{sec:design}

\subsection{Problem}

With user weights $\alpha_k\ge0$, taking as the utility the weighted sum \emph{spectral efficiency} in bit/s/Hz rather than a rate in bit/s,\footnote{The model is monochromatic: it does not price stored energy, hence not the bandwidth over which a fixed network delivers the designed excitation, which the stored energy fixes through the antenna $Q$ \cite{Yaghjian2005Q}. A rate in bit/s would therefore claim a bandwidth the model does not certify.} and taking the loads lossless for concreteness ($\bM=\im\bX$; the lossy case only adds the fixed $\bR_M$ inside $\bT$), the master problem over $\Utt$, $\bX$, and $\bW$ collects the constraints of Section~\ref{sec:loaded}:
\begin{align}
(\mathrm{P1}):\ \max_{\Utt,\bX,\bW}\ &\sum_{k=1}^K\alpha_k\log_2\big(1+\Gamma_k(\Utt,\bX,\bW)\big)\label{eq:P1}\\
\text{s.t.}\ \ \mathrm{C1}:&\ \textstyle\sum_k\bw_k^{\mathsf H}\bC_{\rm in}(\Utt,\bX)\bw_k\le P,\nonumber\\
\mathrm{C2}:&\ \textstyle\sum_k|\be_n^{\mathsf T}\bar\bZ_{\rm in}\bw_k|^2\le\bar V_{\max}^2,\ n=1,\ldots,N_0,\nonumber\\
\mathrm{C3}:&\ \textstyle\sum_k|\be_n^{\mathsf T}\bB\bw_k|^2\le\bar I_{\max}^2,\ n=1,\ldots,N_{\rm s},\nonumber\\
\mathrm{C4}:&\ \bX\in\netset,\nonumber\\
\mathrm{C5}:&\ \bu_{l,n}\in\surf_l,\ \|\bu_{l,n}-\bu_{l,m}\|\ge d_{\min},\ n\ne m.\nonumber
\end{align}
C1--C3 are the constraints of Section~\ref{sec:constraints}, C4 is the interconnection pattern, and C5 keeps every port in its aperture and at least $d_{\min}$ from the other ports of its plane. Coupling shapes both the feasible set and the effective channel, and both move with the positions and the loads. (P1) is non-convex, and we alternate over three blocks whose ingredients are all closed-form.

\subsection{Precoder block}

Fix $\Utt$ and $\bX$. By Lemma~\ref{lem:passivity}, $\bC_{\rm in}\succ0$; whitening $\tilde\bw_k=\bC_{\rm in}^{1/2}\bw_k$, $\tilde\bh_k=\bC_{\rm in}^{-1/2}\tilde\bg_k^*$ leaves $\tilde\bg_k^{\mathsf T}\bw_{k'}=\tilde\bh_k^{\mathsf H}\tilde\bw_{k'}$ invariant and reduces C1 to $\sum_k\|\tilde\bw_k\|^2\le P$. The block is then the canonical MU-MISO weighted sum spectral efficiency problem, solved by the WMMSE iterations of \cite{Shi2011WMMSE} exactly as in \cite[Sec.~IV-A]{Iacovelli2026MC}, with C2 and C3, both second-order cone constraints in $\bW$, enforced by the penalty $\sum_k\big(\sum_n\varsigma_n|\be_n^{\mathsf T}\bar\bZ_{\rm in}\bw_k|^2+\sum_n\varpi_n|\be_n^{\mathsf T}\bB\bw_k|^2\big)$ and projected-subgradient multiplier updates. Each pass is monotone in the weighted sum spectral efficiency.

\subsection{Load block: projected gradient on the interconnection pattern}

Fix $\Utt$ and $\bW$. The rate depends on $\bX$ only through $\bT$ in (\ref{eq:effchannel}), (\ref{eq:Zin}), and (\ref{eq:induced}). Let $\ba_k\triangleq\bT\bg_{k,\rm s}$ and $\bB=\bT\bZ_{\rm s0}$; since $\bZ_{\rm ss}+\im\bX$ is symmetric, so is $\bT$.
\begin{lemma}[Load gradients]\label{lem:loadgrad}
For the symmetric entry $X_{nm}$ ($n\le m$, both $(n,m)$ and $(m,n)$ moving), with $\bE_{nm}=\be_n\be_m^{\mathsf T}+\be_m\be_n^{\mathsf T}$ for $n\ne m$ and $\bE_{nn}=\be_n\be_n^{\mathsf T}$,
\begin{align}
\frac{\partial\tilde\bg_k^{\mathsf T}}{\partial X_{nm}}&=\im\,\ba_k^{\mathsf T}\bE_{nm}\bB=\im\big(a_{k,n}\bB_{m,:}+a_{k,m}\bB_{n,:}\big),\label{eq:dgdX}\\
\frac{\partial\bZ_{\rm in}}{\partial X_{nm}}&=\im\,\bB^{\mathsf T}\bE_{nm}\bB,\qquad
\frac{\partial\bB}{\partial X_{nm}}=-\im\,\bT\bE_{nm}\bB,\label{eq:dZindX}
\end{align}
(for $n=m$ drop the second term in (\ref{eq:dgdX})). The gradient of the weighted sum spectral efficiency follows by the chain rule through (\ref{eq:sinr}) with the rate gradient of \cite[Lemma~3(ii)]{Iacovelli2026MC} evaluated on $\tilde\bg_k$, and the gradients of the penalties of C1-C3 through (\ref{eq:dZindX}).
\end{lemma}
\begin{proof}
$\partial\bT=-\bT(\im\partial\bX)\bT$ with $\partial\bX=\bE_{nm}$; substituting in (\ref{eq:effchannel}), (\ref{eq:Zin}), (\ref{eq:induced}) and using the symmetry of $\bT$ gives the claims.
\end{proof}
The pattern $\netset$ is a linear subspace of $\mathbb S^{N_{\rm s}}$, whatever its layer structure, so its projection $\Pi_{\netset}$ is a mask, and the block update is $\bX^+=\Pi_{\netset}\big(\bX+\gamma_{\bX}\nabla_{\bX}\lagr\big)$. For $\lagr$ we take the objective evaluated with $\bW$ rescaled to C1 at every $\bX$, which is what is scored. At a feasible point its gradient is that of the partial Lagrangian of (P1) with multiplier $S/(2P)$, where $S=2\sum_k\Real\{\tilde\bg_k^{\mathsf T}\bm\zeta_k\}$ is the sensitivity of the objective to a common scaling of the precoders and $\bm\zeta_k$ are the rate-gradient vectors of \cite[Lemma~3(ii)]{Iacovelli2026MC}, so the C1 term enters through (\ref{eq:dZindX}); the step $\gamma_{\bX}>0$ is set by Armijo backtracking on the same function. After one $O(N_{\rm s}^3)$ inversion all $\dim\netset$ derivatives cost $O(KN_{\rm s}^2)$, being assembled once from $\ba_k$ and $\bB$ and then masked, so the cost is the same for every pattern; for a diagonal layer the step is a per-element reactance update, the physically consistent replacement of the per-element phase. For a single user and one fully connected reflective layer global closed forms under coupling exist \cite{Nerini2024closedform} and benchmark the block; the projected gradient trades that global optimality for multiple users, layers, and sparse patterns.

\subsection{Position block: closed-form 3-D gradients}

Fix $\bX$ and $\bW$. Ports move in their planes, so gradients are with respect to $\bu_{l,n}\in\mathbb R^2$, but the kernel they differentiate is three-dimensional.
\begin{lemma}[3-D kernel gradient]\label{lem:kernelgrad}
For two ports with planar separation $\btau$ and axial separation $h$, $R=\sqrt{\|\btau\|^2+h^2}$, $x=\kappa R$, $c=\bpb^{\mathsf T}\btau/R$,
\begin{equation}
\nabla_{\btau}\vartheta(x,c)=\vartheta'(x,c)\,\frac{\kappa\btau}{R}+3c\,h_2(x)\,\frac{\bpb-c\,\btau/R}{R},\label{eq:gradtheta3d}
\end{equation}
with $\vartheta'(x,c)=\tfrac32\big[-h_1-\tfrac{h_0}{x}+\tfrac{3h_1}{x^2}+c^2(h_1-\tfrac{3h_2}{x})\big]$ the $x$-derivative at fixed $c$; at $h=0$, (\ref{eq:gradtheta3d}) reduces to \cite[Eq.~(47)]{Iacovelli2026MC}. The substitution $h_\ell\mapsto j_\ell$ gives $\nabla\rho$, and for a profile of finite size the gradient of (\ref{eq:Zports}) replaces (\ref{eq:gradtheta3d}). Across planes the gradient is bounded, $\|\nabla\vartheta\|\le C(\kappa h)$, since $R\ge h$.
\end{lemma}
\begin{proof}
$\nabla_{\btau}x=\kappa\btau/R$ and $\nabla_{\btau}c=(\bpb-c\,\btau/R)/R$ by direct differentiation of $c=\bpb^{\mathsf T}\btau/\sqrt{\|\btau\|^2+h^2}$; the chain rule with $\partial_c\vartheta=3c\,h_2$ gives (\ref{eq:gradtheta3d}).
\end{proof}

\begin{lemma}[Position gradients through the stack]\label{lem:posgrad}
Let $\partial$ denote the derivative with respect to one coordinate of one port position. With $\ba_k=\bT\bg_{k,\rm s}$ and $\bB=\bT\bZ_{\rm s0}$,
\begin{align}
\partial\tilde\bg_k^{\mathsf T}&=\partial\bg_{k,0}^{\mathsf T}-\partial\bg_{k,\rm s}^{\mathsf T}\bB\nonumber\\
&\quad+\ba_k^{\mathsf T}(\partial\bZ_{\rm ss})\bB-\ba_k^{\mathsf T}\partial\bZ_{\rm s0},\label{eq:dgdu}\\
\partial\bZ_{\rm in}&=\partial\bZ_{00}-\partial\bZ_{0\rm s}\bB\nonumber\\
&\quad-\bB^{\mathsf T}\partial\bZ_{\rm s0}+\bB^{\mathsf T}(\partial\bZ_{\rm ss})\bB,\label{eq:dZindu}
\end{align}
where, if the moving port is the transmit port $n$, only column $n$ of $\bZ_{\rm s0}$, row and column $n$ of $\bZ_{00}$, and entry $n$ of $\bg_{k,0}$ are nonzero in the derivatives. If it is a stack port $n$, only row $n$ of $\bZ_{\rm s0}$, row and column $n$ of $\bZ_{\rm ss}$, and entry $n$ of $\bg_{k,\rm s}$ are. Every nonzero entry is $R_r$ times (\ref{eq:gradtheta3d}) evaluated at the corresponding 3-D separation, and the channel entries follow \cite[Eq.~(51)]{Iacovelli2026MC} with $\bm\tau=\br_k-\br_{l,n}$ three-dimensional.
\end{lemma}
\begin{proof}
Differentiate (\ref{eq:effchannel}) and (\ref{eq:Zin}) using $\partial\bT=-\bT(\partial\bZ_{\rm ss})\bT$ and the symmetry of $\bT$.
\end{proof}

The position update is $\bu_{l,n}^+=\Pi_{\surf_l}\big(\bu_{l,n}+\gamma_{\Utt}\nabla_{\bu_{l,n}}\lagr\big)$ on every plane with a pairwise push-out for C5, on the same rescaled objective and with the same Armijo rule. The sparsity of Lemma~\ref{lem:posgrad} makes one port's gradient cost $O(KN_{\rm s}+N_{\rm s})$ after the inversion, so a full sweep over all ports costs $O(KN_{\rm s}^2)$, the same order as the load block.

\subsection{Certificates, initialization, algorithm, and complexity}\label{sec:cert}

The certificate relaxes the transmit plane, i.e., the codeword map (\ref{eq:qmap}). Under stream $k$ that plane carries $\biota_0=\bw_k$, hence the modal coefficients $\tilde\bq_k\triangleq\bq_0=\bF\bw_k$ with the transmit codeword map $\bF\triangleq\bXi\bPhi_0$, and the holographic upper bound (HUB) of \cite[Sec.~IV-C]{Iacovelli2026MC} lets $\tilde\bq_k$ range over all of $\mathbb C^{N_{\Iset}}$ rather than over the codeword-generated set $\{\bF\bw:\bw\in\mathbb C^{N_0}\}$ a port configuration can reach. The stack is not relaxed: it enters through the two mixed blocks obtained by stripping $\bF$ off the couplings that leave and reach the transmit plane,
\begin{equation}
\bZ^{\rm m}_{\rm s0}\triangleq\bPhi_{\rm s}^{\mathsf H}\bXi_{\rm s}\bLam^{(\rm s0)},\quad \bZ^{\rm m}_{0\rm s}\triangleq\bLam^{(0\rm s)}\bXi_{\rm s}\bPhi_{\rm s},\quad \bB^{\rm m}\triangleq\bT\bZ^{\rm m}_{\rm s0},\label{eq:mixedblock}
\end{equation}
with $\bPhi_{\rm s}\triangleq\blkdiag(\bPhi_1,\ldots,\bPhi_L)$, $\bXi_{\rm s}\triangleq\bI_L\otimes\bXi$, and $\bLam^{(\rm s0)}$ the column of blocks $\bLam^{(l0)}$.

\begin{proposition}[Holographic bound in front of a stack]\label{prop:hub}
Read on the modes, the induced currents of (\ref{eq:induced}) are $\biota_{\rm s}=-\bB^{\rm m}\tilde\bq_k$, and the effective channel of (\ref{eq:effchannel}) and the delivered power of Lemma~\ref{lem:passivity} are carried by
\begin{equation}
(\tilde\bg_k^{\rm m})^{\mathsf T}=(\hat\bg_k^{(0)})^{\mathsf T}-\bg_{k,\rm s}^{\mathsf T}\bB^{\rm m},\quad R_r\bC^{\rm m}_{\rm in}=\Herm\big\{\bLam^{(00)}-\bZ^{\rm m}_{0\rm s}\bB^{\rm m}\big\},\label{eq:gm}
\end{equation}
$\Herm\{\bA\}=\tfrac12(\bA+\bA^{\mathsf H})$, in the sense that they return $\tilde\bg_k^{\mathsf T}\bw$ and $\bw^{\mathsf H}\bC_{\rm in}\bw$ whenever $\tilde\bq=\bF\bw$. Consequently the value of
\begin{align}
(\mathrm{P2}):\ \max_{\{\tilde\bq_k\}}\ &\sum_{k=1}^{K}\alpha_k\log_2\!\Big(1+\frac{|(\tilde\bg_k^{\rm m})^{\mathsf T}\tilde\bq_k|^2}{\sum_{k'\ne k}|(\tilde\bg_k^{\rm m})^{\mathsf T}\tilde\bq_{k'}|^2+\sigma_k^2}\Big)\nonumber\\
\text{s.t.}\ &\textstyle\sum_k\tilde\bq_k^{\mathsf H}\bC^{\rm m}_{\rm in}\tilde\bq_k\le P,\label{eq:P2}
\end{align}
bounds (P1) from above for every transmit configuration at fixed stack geometry and loads.
\end{proposition}
\begin{proof}
See \apx{app:hub}{H}.
\end{proof}

These factorizations are exact in the diagonal limit of Lemma~\ref{lem:diag3d}. At finite aperture we evaluate $\bZ^{\rm m}_{\rm s0}$, $\bZ^{\rm m}_{0\rm s}$, and the resistive part of $\bLam^{(00)}$ against the modes themselves rather than their lattice samples, which leaves only the truncation of $\Iset$ as error, and report the residual of the identity. The flat ohmic part of $\bLam^{(00)}$, scaled to reproduce $\varepsilon$ on the port diagonal, keeps $\bC^{\rm m}_{\rm in}$ positive definite on the evanescent modes. Algorithm~\ref{alg:ao} evaluates (P2) at the configuration it returns; its solution, projected onto the codeword manifold by matching pursuit over all $K$ streams at once, seeds the transmit positions. C2 is dropped, a port voltage having no field-domain meaning, and so is C3, whose field counterpart would cost the closed-form WMMSE structure; dropping constraints keeps the bound valid. On the layers, the fully connected pattern upper-bounds every sparser pattern at fixed geometry. The layer positions are initialized on the $\lambda/2$ grid, where the in-plane resistive coupling nearly vanishes \cite[Cor.~1]{Iacovelli2026MC}. Loads start from the best of a few uniform reactances at the two extremes of the terminal law (\ref{eq:terminal}): every atom at its own resonance, $\bX=-\Imag\{\diag(\bZ_{\rm ss})\}$, which draws the largest currents and dissipation (Corollary~\ref{cor:eff}), and the open circuit, $\bX=\diag\{X_{l,n}\}$ with every $|X_{l,n}|\to\infty$, where $\bT\to\mathbf0$, which is transparent and has no gradient, so the returned point is compared with that inert stack.

\begin{algorithm}[t]
\caption{Coupling-aware design of the stacked fluid metasurface}
\label{alg:ao}
\begin{algorithmic}[1]
\STATE \textbf{Input:} apertures $\surf_l$, heights $z_l$, $\bp$, $\varepsilon$, $\zeta_A$, pattern $\netset$, $P$, $\bar V_{\max}$, $\bar I_{\max}$, $d_{\min}$, $\alpha_k$, user spectra on every plane
\STATE initialize $\Utt^{(0)}$ (plane $0$: matching pursuit on the solution of (P2); planes $l\ge1$: $\lambda/2$ grid), $\bX^{(0)}$ (best uniform reactance, Section~\ref{sec:cert}), $\bW^{(0)}$ (whitened equal-power MRT)
\REPEAT
\STATE build $\bZ(\Utt)$ by (\ref{eq:Znorm}); invert $\bT$; form $\bB$, $\bZ_{\rm in}$, $\bC_{\rm in}$, $\{\tilde\bg_k\}$
\STATE \textit{precoder block:} whiten by $\bC_{\rm in}$, WMMSE passes with C2-C3 penalties
\STATE \textit{load block:} $T_X$ projected-gradient steps on $\bX$ with Lemma~\ref{lem:loadgrad}
\STATE \textit{position block:} $T_U$ projected-gradient steps on every plane with Lemmas~\ref{lem:kernelgrad}-\ref{lem:posgrad}; push-out for C5
\UNTIL{objective change $<$ tolerance}
\STATE \textbf{Output:} $\Utt$, $\bX$, $\bW$, and the certificate gaps to the HUB and to the fully connected anchor
\end{algorithmic}
\end{algorithm}

Each block is monotone, the precoder block by \cite{Shi2011WMMSE} and the other two by the Armijo rule, and the objective is bounded, so the objective sequence converges and standard arguments yield a partial stationary point of (P1). An outer iteration costs one $O(N_{\rm s}^3)$ inversion, one $O(N_0^3)$ whitening, $O(K^2N_0+KN_0^2)$ per WMMSE pass, and $O(KN_{\rm s}^2)$ per load or position sweep, with no full-wave solver and no nested cascade. For spectrally local layers, Corollary~\ref{cor:ladder} replaces the inversion by $N_{\Iset}$ tridiagonal solves of size $L$, $O(N_{\Iset}L)$ in total and independent of the layer port count.

\section{Numerical Results}\label{sec:numerics}

Whatever model a scheme is \emph{designed} with, it is \emph{scored} under the exact one, with the channel (\ref{eq:effchannel}) and the power of Lemma~\ref{lem:passivity}, and measured against the direct link, obtained by open-circuiting every atom, $\bT=\mathbf0$, so that $\tilde\bg_k=\bg_{k,0}$. We check the convergence of Algorithm~\ref{alg:ao} first, then map the regimes of the stack, and compare designs in spectral efficiency last.

\subsection{Scenario, conventions, and schemes}\label{sec:setup}

Lengths are in wavelengths and impedances are normalized by $R_r$. The elements are Gaussian, $\xi(\bsb)\propto e^{-\|\bsb\|^2/(2\sigma^2)}$ with $\sigma=\lambda/20$, and every entry of $\bZ$ is evaluated from (\ref{eq:Zports}), not from its point limit, by tabulating the order-zero and order-two Hankel transforms of (\ref{eq:zetahat3d}) weighted by $\hat\xi^2$, which carry the $\cos2\varphi$ dependence on the angle $\varphi$ between $\btau$ and $\bpb$, together with the derivatives that give the gradients of Lemma~\ref{lem:kernelgrad}. The diagonal is $\rho_\sigma(0)+\varepsilon+\im\zeta_A$ with $\rho_\sigma(0)=0.943$, $\varepsilon=0.05$, and $\zeta_A=0.58$, and the channels carry the element taper. Ports keep $d_{\min}=0.15\lambda$ and $3\sigma$ from the aperture edge.

$L=2$ sheets of $5\times5$ atoms on a $\lambda/2$ grid share the $4\lambda\times4\lambda$ aperture of the transmit plane, the first at the standoff $z_1-z_0=\lambda/2$ and the two $h=\lambda/4$ apart. The $N_0=4$ transmit ports start on the central $\lambda/2$ square. $K=3$ users with random polarization, handsets of unknown orientation, are dropped uniformly within a lateral half-width of $4\lambda$, $8\lambda$ to $14\lambda$ beyond the last sheet, inside the radiating near field of the aperture, and their channels are normalized to unit mean transmit-port gain, so that $P/\sigma^2=10$~dB is the SNR of the direct link. Curves average eight drops, each with the same seed in every experiment and at every sweep point, so sweeps are paired and figures with the same scenario share the users.

Every design is rescaled to meet C1 with equality under the exact $\bC_{\rm in}$ before it is scored, and keeps any point of its own feasible set that another run has found and that scores higher: the inert stack, which is feasible but cannot start the load block (Section~\ref{sec:cert}), and, in the nested comparisons below, the point of the smaller design. Lattice cells cut by the light circle, where the $1/k_z$ singularity of (\ref{eq:zetahat3d}) is integrable but not samplable, are cell-averaged on an $8\times8$ sub-grid \cite{Iacovelli2026VA}.

MA is Algorithm~\ref{alg:ao} in full, the transmit ports moving continuously as movable antennas do. FAS keeps them on a grid of pitch $\lambda/8$, as the ports of a fluid antenna, at the positions of the (P2) projection that also seeds MA, and optimizes loads and precoders; MIMO does the same with the ports frozen on the $\lambda/2$ square. The square lies on the grid, and grid ports are feasible for MA, so MA$\,\ge\,$FAS$\,\ge\,$MIMO in every drop. Two designs run the MIMO pipeline under one modeling error each: the norm power model (NPM), $\bC_{\rm in}\to\bI_{N_0}$, blind to the pull of Lemma~\ref{lem:passivity}, and the cascade-approximated stack (CAS), under (A1)+(A2) of Proposition~\ref{prop:cascade}; both are scored as designed, since their errors are what they show. DIRECT open-circuits the stack, and HUB is (P2) of Proposition~\ref{prop:hub} at the configuration MA returns, on the modes $\|\bk_i\|\le2.5\kappa$, with the finite-aperture blocks of Section~\ref{sec:cert}. Unless stated, every sheet is fully connected.

\subsection{Convergence}

Fig.~\ref{fig:conv} runs Algorithm~\ref{alg:ao} in the default scenario from three initializations, every curve rising since each block is monotone (Section~\ref{sec:cert}). From the (P2) projection of Proposition~\ref{prop:hub} the design starts at $9.63$~bit/s/Hz, $1.1$ above the $\lambda/2$ square, reaches $99\%$ of its limit after seven outer iterations, one before the square, and ends at $13.15$, within $0.04$ of it; a random constellation ends $0.9$ lower. The projection buys speed and robustness to the start rather than a better point, and eight outer iterations suffice below.

\begin{figure}[t]\centering
\includegraphics[width=.7\columnwidth]{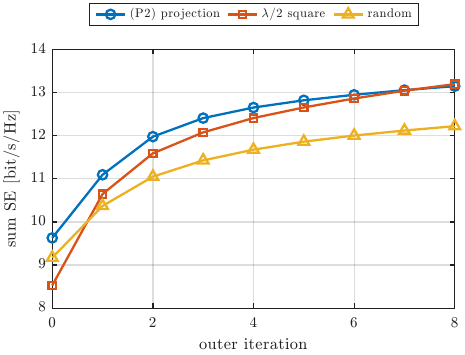}
\caption{Convergence of Algorithm~\ref{alg:ao} from three initializations, default scenario, averaged over the drops.}\label{fig:conv}
\end{figure}

\begin{figure}[t]\centering
\includegraphics[width=.7\columnwidth]{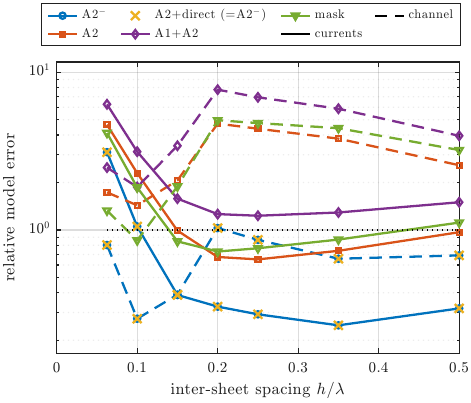}
\caption{Relative model error of the resonant stack against the spacing, on the induced currents (solid) and on the channel (dashed); the crosses, A2+direct, coincide with A2$^-$ at $L=2$.}\label{fig:stack}
\end{figure}

\subsection{What the cascade costs}

Fig.~\ref{fig:stack} prices the assumptions of Proposition~\ref{prop:cascade} one at a time, at resonant loads, together with the unit-modulus phase mask $\diag(e^{\im\theta_{ln}})$ of Proposition~\ref{prop:cascade}(iii) on the scalar propagator of Remark~\ref{rem:RS}, scaled to the port kernel at broadside. The models err most where the exact stack chokes itself (Remark~\ref{rem:stiffness}): at $h=\lambda/16$ the relative error on the induced currents is $3.1$ for (A2$^-$), $4.7$ for (A2), $6.3$ for (A1)+(A2), and $4.1$ for the mask, and every reduced model is still above $1$ at $\lambda/10$. Beyond $0.2\lambda$ the (A2$^-$) error settles between $0.25$ and $0.33$, the backward blocks it drops being the forward ones transposed (Theorem~\ref{thm:cfe3d}(iii)), whereas the (A2) error stays between $0.65$ and $0.97$, because the direct blocks it drops stay strong, $\|\bZ_{20}\|/\|\bZ_{10}\|$ being $0.80$ at $h=\lambda/2$ and $0.66$ at $h=\lambda$; restoring them, the crosses, puts (A2) exactly on (A2$^-$), as Remark~\ref{rem:stiffness} requires at $L=2$. A precoder designed under a reduced model and scored exactly loses up to $1.71$~bit/s/Hz under (A2), $1.74$ under (A1)+(A2), $1.66$ under the mask, and $0.66$ under (A2$^-$), every maximum at $h=\lambda/2$, the widest spacing tested: the models err most in the currents where the stack is choked, but cost most in rate where it is active.

\subsection{How to wire a stack}

The patterns of Section~\ref{sec:loaded} are compared at $L=4$, where the ladder ($2L-1=7$ freedoms per site) and the full per-column $L$-port ($10$) differ. In increasing $\dim\netset$, from $100$ to $5050$ on the axis of Fig.~\ref{fig:wiring}, they are diagonal layers, $\bX$ diagonal (diag); in-plane pairs, each $\bX_{ll}$ block-diagonal in $2\times2$ blocks (group); the vertical ladder, $\bX_{ll}$ and $\bX_{l,l+1}$ diagonal (vertical); an in-plane nearest-neighbour chain per sheet, each $\bX_{ll}$ tridiagonal (tree); all pairs within a column, every $\bX_{ll'}$ diagonal (vertFull); fully connected layers, each $\bX_{ll}$ dense and $\bX_{ll'}=\mathbf0$ for $l\neq l'$ (fullLayer); and the stack-wide network, $\bX$ dense (stackwide). In-plane links follow the serpentine order of the lattice, and the vertical patterns carry one value per site, so they are the supports of the networks of Corollary~\ref{cor:vertical} rather than those networks; the load block of Lemma~\ref{lem:loadgrad} costs the same for all of them.

Fig.~\ref{fig:wiring} (top) shows the rate growing with $\dim\netset$ at both spacings, but not in the order of the axis alone. At comparable dimension the in-plane chain beats the vertical ladder, $9.12$ against $8.76$~bit/s/Hz at $h=0.125\lambda$ and $12.97$ against $12.44$ at $0.375\lambda$, while the full per-column $L$-port beats both, with $9.54$ and $13.39$: vertical connectivity pays when each column is fully connected, not when it is a ladder. The dense patterns keep improving, fullLayer reaching $11.51$ and $14.30$ and stackwide $12.50$ and $15.59$, with $5$ to $29$ times the reactances of the column patterns. Widening the gaps from $0.125\lambda$ to $0.375\lambda$ adds $2.8$ to $3.9$~bit/s/Hz to every pattern, the choking of Remark~\ref{rem:stiffness} at $L=4$. Every pattern reached at least the rate of the patterns nested in it on its own, without the fallback of Section~\ref{sec:setup}.

Fig.~\ref{fig:wiring} (bottom) crosses fixed or movable ports, the latter moving the transmit ports and the atoms of both sheets by Lemma~\ref{lem:posgrad}, with layer-diagonal or fully connected sheets. Here the users are dropped once, at the widest spacing, so both references are flat in $h$. Connecting is the better purchase between $0.15\lambda$ and $0.35\lambda$, by up to $0.43$~bit/s/Hz, and moving at the ends, by $0.62$ at $\lambda/16$, $0.29$ at $\lambda/10$, and $0.43$ at $\lambda/2$; over that range their combination keeps $81\%$ to $89\%$ of the sum of the two gains, both reshaping the one operator $\bZ_{\rm ss}+\bM$ that (\ref{eq:induced}) inverts. Every two-sheet design peaks at $h=0.35\lambda$, beyond the spacing $h\approx\sqrt2\lambda/(2\pi)\approx0.225\lambda$ at which the resonant stack draws the most current, where $\chi_{\rm ev}$ changes sign (Corollary~\ref{cor:axial}), on a plateau from $0.15\lambda$ over which the movable connected design stays within $1.7\%$ of its $13.94$~bit/s/Hz. There the second sheet adds $1.20$ to $1.44$~bit/s/Hz to a single movable connected sheet, and at most $0.38$ elsewhere; at $\lambda/16$ the best is always the single sheet with the second one open.

\begin{figure}[t]\centering
\includegraphics[width=.7\columnwidth]{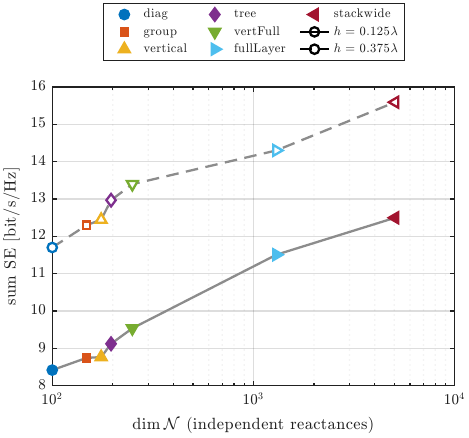}\\[2pt]
\includegraphics[width=.7\columnwidth]{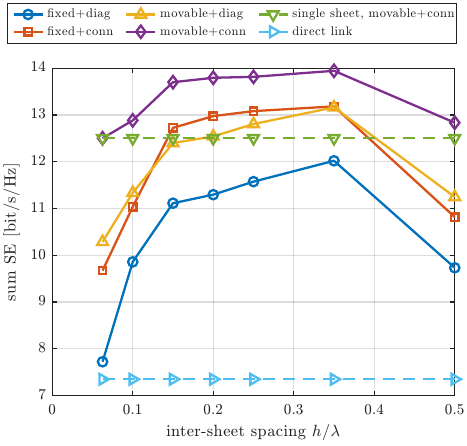}
\caption{How to wire a stack. Top: architecture ladder against the interconnection dimension, $L=4$; filled markers $h=0.125\lambda$, open markers $h=0.375\lambda$. Bottom: moving, interconnecting, or both.}\label{fig:wiring}
\end{figure}

\subsection{How deep, and how many ports}

By (\ref{eq:effchannel}) the effective channel of every user is a row of length $N_0$ whatever the depth, so sheets cannot add streams, only improve the $K$-user channel. Fig.~\ref{fig:scale} (top) sweeps $L$ for $N_0=K\in\{2,4\}$, with users within a lateral half-width of $4\lambda$ (spread) or $\lambda$ (clustered), on sheets of $4\times4$ atoms; users are dropped once, beyond the deepest stack, and each depth keeps the previous design, with the new sheet open, when that is better. Depth pays more with more streams, and not more for close users: six sheets double the rate with four streams, from $7.56$ to $15.48$~bit/s/Hz with spread users and from $5.05$ to $10.38$ with clustered ones, and raise it by $82\%$ and $60\%$ with two. With four streams $88\%$ of the gain is in place at $L=3$, and each further sheet adds at most $0.44$~bit/s/Hz. As $L$ is fixed at build time, this makes $L=3$ the robust choice for four streams, whatever the users' layout.

Fig.~\ref{fig:scale} (bottom) prices a two-sheet stack in ports, with $K=4$ clustered users, the spacing of maximum current, $h=0.225\lambda$, and nested frozen arrays, the first $N_0$ cells of a $\lambda/2$ lattice ordered outward from the centre. The stack pays at every port count: with frozen ports it adds $2.7$ to $3.9$~bit/s/Hz for $N_0\ge2$, and one port behind it, at $7.22$~bit/s/Hz, beats eight bare ports, at $6.33$. Positioning adds $1.7$ to $4.4$~bit/s/Hz on top for $N_0\ge2$, growing with $N_0$, and nothing for a single port, best left at the centre. The $\lambda/8$ grid keeps almost all of it, a fluid port being the codeword of Proposition~\ref{prop:codeword3d} sampled on a grid: FAS trails MA by at most $0.36$~bit/s/Hz ($3\%$) up to $N_0=4$, and by $0.74$ and $1.10$ ($6\%$ and $8\%$) at $N_0=6$ and $8$, where the ports crowd and MA moves them $0.17\lambda$ on average. With loads and precoders optimized for both, the two position sets thus perform alike, unlike in \cite{Iacovelli2026MC}, where the fluid design stopped at the projection, which makes the fluid grid the practical choice over mechanical motion. The same runs price the modeling errors: ignoring the pull of Lemma~\ref{lem:passivity} (NPM) costs at most $1.21$~bit/s/Hz, at $N_0=2$, and nothing at $N_0=4$, whereas designing with the cascade of Proposition~\ref{prop:cascade} (CAS) costs $1.45$ at $N_0=1$ and $9.52$ at $N_0=8$ and, from $N_0=4$ on, falls below the direct link. That loss is in the loads, tuned for currents the true stack does not carry (Remark~\ref{rem:stiffness}): redesigning the precoder on the truth at the same loads recovers $0.06$~bit/s/Hz at $N_0=1$ and $1.34$ at $N_0=8$, leaving $1.39$ and $8.18$ to the loads, which keep the stack below the direct link from $N_0=6$ on. Swept in SNR instead, the same stack returns the certificates: MA reaches $48\%$ of the bound of Proposition~\ref{prop:hub} at $0$~dB and $70\%$ at $20$~dB, at efficiency (\ref{eq:eta}) $0.92$ to $0.95$, clear of the resonant point where Corollary~\ref{cor:eff} makes control and dissipation coincide, with C2 and C3 inactive at a peak port voltage of $7.2$ and a peak induced current of $0.56$, and the identity behind the bound holds to within $12\%$ in power, its rate residual growing to $0.81$~bit/s/Hz at $20$~dB.

\begin{figure}[t]\centering
\includegraphics[width=.7\columnwidth]{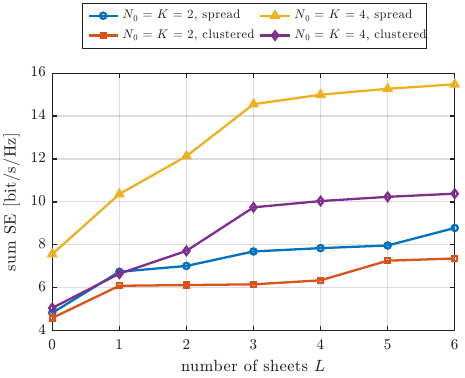}\\[2pt]
\includegraphics[width=.7\columnwidth]{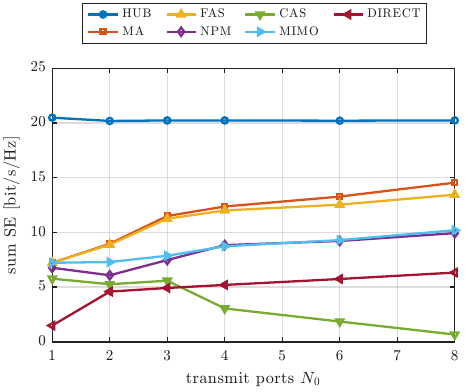}
\caption{How deep, and how many ports. Top: returns of depth, $N_0=K$, spread and clustered users. Bottom: a two-sheet stack priced in transmit ports, $K=4$ clustered users, nested frozen arrays.}\label{fig:scale}
\end{figure}

\section{Conclusion}

A stacked fluid metasurface was modeled on parallel current sheets with only the transmit sheet driven. The impedance kernel at three-dimensional separations is the coupling matrix of the whole constellation, and one inversion gives the induced currents, the multi-user channel, and the stack-loaded input impedance, differentiable in every position and every load. The cascade of the literature is its single-pass limit under two assumptions that fail where stacks are most compact, and the price of a stack is the induced-current norm its sheets dissipate at the transmitter's expense. In the wavenumber domain every mode sees a loaded transmission line, a passive sheet transmits one plus what it reflects, and vertical networks repeated at every site are spectrally local. Numerically, designing with the cascade can leave a stack below the direct link, depth pays more with more streams, and fluid ports on a $\lambda/8$ grid come within $3\%$ of continuously movable ones, so the mechanics of movable antennas buy little in a stacked transmitter. Future work includes bandwidth accounting for resonant layers, coupling-aware acquisition of the per-plane spectra \cite{Iacovelli2026VA}, dielectric spacers, and vertical connections of non-negligible electrical length.

\bibliographystyle{IEEEtran}
\bibliography{refs}

\ifwithappendix
\appendices

\section{Proof of Theorem~\ref{thm:cfe3d}}\label{app:cfe3d}
(i) Substituting $\bj=\sum_{l,n}\iota_{l,n}\xi(\cdot-\bu_{l,n})\bp\,\delta(z-z_l)$ into (\ref{eq:Pcx}) and exchanging the finite sums with the integrals gives $P_{\rm cx}=\tfrac12\sum\iota^*_{l,n}\iota_{l',m}\langle\xi_{l,n}\bp,\kerZ\xi_{l',m}\bp\rangle$; the sifting property of the two deltas reduces the volume integrals to the two planes, on which $\kerZ$ depends on the arguments only through $\bsb-\bsb'$ and $z_l-z_{l'}$, and the substitution $\bsb=\bu_{l,n}+\bm\alpha$, $\bsb'=\bu_{l',m}+\bm\beta$ yields the correlation form (\ref{eq:Zports}); the conjugate on $\xi$ drops because the profile is real, which makes $\bZ$ complex symmetric. Uniqueness and the descent of the split follow as in \cite[App.~C]{Iacovelli2026MC}, since the argument uses only the reality and symmetry of $\kerR$, $\kerX$, which hold for any set of planes.
(ii) For distinct ports the kernel is smooth on the integration region: on the same plane as soon as the supports are disjoint, across planes for every $\btau$ since $R\ge h>0$. Expanding around the separation, the zeroth order contributes $\ell_{\rm e}^2\zeta_p(\btau;h)$, the first vanishes by evenness of $\xi$, and the second is smaller by $O\big((\kappa\Delta)^2+(\Delta/R)^2\big)$, which is why the limit needs $\Delta\ll R$ as well as $\kappa\Delta\ll1$ (Section~\ref{sec:model}). Evaluating $-\im\kappa Z_0\ell_{\rm e}^2\bp^{\mathsf T}\bG(\bm\tau)\bp$ with the Hankel form of \cite[Eq.~(4)]{Iacovelli2026MC} at the three-dimensional $\bm\tau=[\btau^{\mathsf T},h]^{\mathsf T}$ gives $R_r\vartheta(\kappa R,\bp^{\mathsf T}\hat{\bm\tau})$, and $\bp^{\mathsf T}\hat{\bm\tau}=\bpb^{\mathsf T}\btau/R$ because $\bp\perp\hat{\mathbf z}$. The self-terms are on a single plane and are those of \cite[Thm.~1(iii)-(iv)]{Iacovelli2026MC}.
(iii) Symmetry is reciprocity: $\bG(\bm\tau)=\bG(-\bm\tau)^{\mathsf T}$ and the profiles are real. \hfill$\blacksquare$

\section{Proof of Lemma~\ref{lem:passivity}}\label{app:passivity}
The Hermitian part of $\bZ_{\rm ss}+\bM$ is $\Real\{\bZ_{\rm ss}\}+\bR_M$. Since $\Real\{\bZ_{\rm ss}\}=R_r(\varepsilon\bI_{N_{\rm s}}+\bC_{\rm ss})$ with $\bC_{\rm ss}\succeq0$ the port-sampled radiation kernel of the stack ports, and $\bR_M\succeq0$, the Hermitian part is $\succeq\varepsilon R_r\bI_{N_{\rm s}}$ and the matrix is invertible. Eliminating $\biota_{\rm s}$ by (\ref{eq:induced}) from the first row of (\ref{eq:multiport}) gives (\ref{eq:Zin}), symmetric as a Schur complement of a symmetric matrix. For the power balance, the real power crossing all port planes is, by (\ref{eq:balance}), $\sum_lP_{\Omega,l}+P_{\rm rad}=\tfrac12\Real\{\biota^{\mathsf H}\bZ\biota\}$ with $\biota=[\biota_0;\biota_{\rm s}]$ and $P_{\Omega,l}=\tfrac{\varepsilon R_r}{2}\|\biota_l\|^2$ on every sheet, driven or not. The generators at the stack ports are the load networks, which deliver $\tfrac12\Real\{\biota_{\rm s}^{\mathsf H}\bv_{\rm s}\}=-\tfrac12\Real\{\biota_{\rm s}^{\mathsf H}\bM\biota_{\rm s}\}=-\tfrac12\biota_{\rm s}^{\mathsf H}\bR_M\biota_{\rm s}=-P_M$, i.e., they absorb $P_M$. Hence the transmit generators supply $\tfrac12\Real\{\biota_0^{\mathsf H}\bv_0\}=\tfrac12\biota_0^{\mathsf H}\Real\{\bZ_{\rm in}\}\biota_0=\sum_lP_{\Omega,l}+P_M+P_{\rm rad}$, which is (\ref{eq:power3d}). Explicitly, expanding $\biota^{\mathsf H}\bZ\biota$ with $\biota_{\rm s}=-\bB\biota_0$ and using $(\bZ_{\rm ss}+\bM)\bB=\bZ_{\rm s0}$ in the stack block gives $\biota_0^{\mathsf H}(\bZ_{00}-\bZ_{0\rm s}\bB-\bB^{\mathsf H}\bM\bB)\biota_0$, whose Hermitian part is $\Real\{\bZ_{\rm in}\}-\bB^{\mathsf H}\bR_M\bB$, the reactive part $\im\bB^{\mathsf H}\bX\bB$ of the last term being anti-Hermitian. Thus $\Real\{\bZ_{00}\}$ prices the transmit sheet alone, $-\Real\{\bZ_{0\rm s}\bB\}$ everything the stack does, its radiation, its conduction loss, and the interference of its field with the transmitter's, which admit no finer split, and $\bB^{\mathsf H}\bR_M\bB$ the loads. Every term on the right of (\ref{eq:power3d}) is non-negative and $P_{\Omega,0}=\tfrac{\varepsilon R_r}{2}\|\biota_0\|^2$ for every $\biota_0$, so $\bC_{\rm in}\succeq\varepsilon\bI_{N_0}$. \hfill$\blacksquare$

\section{Proof of Lemma~\ref{lem:sheet}}\label{app:star}
By the Weyl identity, the field radiated by tangential currents on $z=z_l$ has, at any $z\ne z_l$, the planar spectrum $\propto(\bI_3-\mathbf k_\pm\mathbf k_\pm^{\mathsf T}/\kappa^2)\hat{\bj}(\bk)\,e^{\im k_z|z-z_l|}/k_z$ with $\mathbf k_\pm=[\bk^{\mathsf T},\pm k_z]^{\mathsf T}$, the sign following that of $z-z_l$. For tangential $\hat{\bj}$ the transverse-transverse block of the projector is $\bI_2-\bk\bk^{\mathsf T}/\kappa^2$ on both sides, so the transverse scattered spectrum is the same function of $|z-z_l|$ above and below the sheet: the sheet scatters symmetrically, mode by mode. Whatever the loads and the interconnection, the induced currents are tangential and produce this symmetric field of amplitude $\beta$. On the transmission side the total transverse field is the incident plus the scattered wave, $t=1+\beta$, and on the reflection side it is the scattered wave alone, $r=\beta$, which is (\ref{eq:star}). For a lossless sheet, per-mode power conservation, $|t|^2+|r|^2=1$, combined with $t-r=1$ gives $\Real\{r\}=-|r|^2$, whose solutions are $r=\tfrac12(e^{\im\psi}-1)$. At the operator level the same argument gives $\bT_{\rm sh}=\bI_{N_{\Iset}}+\bR_{\rm sh}$, losslessness forces $\bR_{\rm sh}=\tfrac12(\bU-\bI_{N_{\Iset}})$ with $\bU$ unitary, and reciprocity makes $\bU$ complex symmetric, so its real and imaginary parts are commuting real symmetric matrices and $\bU=\bQ\diag(e^{\im\psi_i})\bQ^{\mathsf T}$ with $\bQ$ real orthogonal. A two-sheet layer scatters the superposition of two sheets with a phase delay between them, the symmetric argument no longer applies to the pair, and with a network across the sheets the pair realizes an arbitrary lossless two-port per mode, in particular the Huygens condition $r=0$, $|t|=1$, $\arg t$ free. \hfill$\blacksquare$

\section{Proof of Theorem~\ref{thm:dichotomy3d}}\label{app:dichotomy3d}
The Weyl identity \cite[Ch.~2]{Chew1995} expands the scalar spherical wave into plane waves over the aperture plane: for $z\ne0$, $e^{\im\kappa\|\bm\tau\|}/(4\pi\|\bm\tau\|)=\tfrac{\im}{8\pi^2}\int e^{\im(\bk\cdot\btau+k_z|z|)}k_z^{-1}d\bk$. Applying $(\bI_3+\kappa^{-2}\nabla\nabla^{\mathsf T})$ under the integral turns each plane wave into $(\bI_3-\mathbf k_\pm\mathbf k_\pm^{\mathsf T}/\kappa^2)e^{\im(\bk\cdot\btau+k_z|z|)}$, and for tangential $\bp$ the sandwich $\bp^{\mathsf T}(\cdot)\bp=1-(\bk\cdot\bpb)^2/\kappa^2$ is independent of the hemisphere. Hence the planar transform of $-\im\kappa Z_0\bp^{\mathsf T}\bG(\btau,h)\bp$ is $\kappa Z_0(1-c_{\bk}^2)e^{\im k_zh}/(2k_z)$ for $h>0$; for $h=0$ the limit $z\to0^\pm$ is unambiguous for the transverse block and adds the flat transform of the ohmic delta, as in \cite[App.~D]{Iacovelli2026MC}, which is (\ref{eq:zetahat3d}). The split is pointwise: on the disc $k_z>0$ is real, so $e^{\im k_zh}/(2k_z)$ has real part $\cos(k_zh)/(2k_z)$ and imaginary part $\sin(k_zh)/(2k_z)$; outside, $k_z=\im|k_z|$ turns the factor into $-\im e^{-|k_z|h}/(2|k_z|)$, purely imaginary and decaying. \hfill$\blacksquare$

\section{Proof of Corollary~\ref{cor:axial}}\label{app:axial}
Put $\varrho=\|\bk\|$ and specialize Theorem~\ref{thm:dichotomy3d} to facing ports, $\btau=\mathbf0$, where the plane-wave factor $e^{\im\bk\cdot\btau}$ is unity and $\zeta_p(\mathbf0;h)=\frac{1}{4\pi^2}\int\hat\zeta_p(\bk;h)\,d\bk$. Write that integral in polar coordinates, $d\bk=\varrho\,d\varrho\,d\varphi$. The azimuth enters only through $c_{\bk}=\bk^{\mathsf T}\bpb/\kappa=(\varrho/\kappa)\cos\varphi$, whose square averages to $\varrho^2/(2\kappa^2)$ over the circle of radius $\varrho$, so integrating $\varphi$ out leaves the radial weight $1-\varrho^2/(2\kappa^2)$ and the Jacobian $\varrho$, $\zeta_p(\mathbf0;h)=\frac{\kappa Z_0}{4\pi}\int_0^\infty\big(1-\frac{\varrho^2}{2\kappa^2}\big)\frac{e^{\im k_zh}}{k_z}\varrho\,d\varrho$, and dividing by $R_r=Z_0\kappa^2/(6\pi)$ turns the prefactor into $3/(2\kappa)$. The disc and the ring are then integrated apart, in both cases by differentiating $k_z^2=\kappa^2-\varrho^2$. On the disc, $\varrho<\kappa$, where $k_z$ is real, it gives $\varrho\,d\varrho=-k_z\,dk_z$ with $k_z\in(0,\kappa]$. Since $\varrho:0\to\kappa$ runs $k_z:\kappa\to0$, the minus sign of the substitution is cancelled by $-\int_\kappa^0=\int_0^\kappa$, leaving $\frac{3}{2\kappa}\int_0^\kappa\big(\frac12+\frac{k_z^2}{2\kappa^2}\big)e^{\im k_zh}\,dk_z$, whose real part integrates by parts to $\rho(0,h)$, the radiation thus coming from the disc alone, and whose imaginary part is $\chi_{\rm pr}$ of (\ref{eq:axialsplit}). On the ring, $\varrho>\kappa$, where $k_z=\im\nu$ with $\nu=\sqrt{\varrho^2-\kappa^2}\ge0$, it gives $\varrho\,d\varrho=\nu\,d\nu$ and leaves $-\im\frac{3}{2\kappa}\int_0^\infty\big(\frac12-\frac{\nu^2}{2\kappa^2}\big)e^{-\nu h}\,d\nu$, purely imaginary, and the elementary integrals $\int_0^\infty e^{-\nu h}d\nu=1/h$ and $\int_0^\infty\nu^2e^{-\nu h}d\nu=2/h^3$ give $\chi_{\rm ev}$. The two parts sum to $\chi(0,h)$, and the sign change at $x=\sqrt2$, the $3/(4x)$ decay, and the shares quoted in the corollary follow by evaluating (\ref{eq:axialsplit}) numerically. \hfill$\blacksquare$

\section{Proof of Lemma~\ref{lem:diag3d}}\label{app:diag3d}
(i) The proof of \cite[Lemma~2]{Iacovelli2026MC} uses three properties of the weight: smoothness on the disc away from the rim, integrability on the rim ring, and vanishing outside. The weight $\hat r_p(\cdot\,;h)$ has the first two, with $\cos(k_zh)$ bounded, so the same three-region argument, with $\omega_h$ the modulus of continuity at separation $h$, gives the rate. The imaginary part on the disc, $\propto\sin(k_zh)/k_z$, is bounded at the rim, hence smoother than the resistive weight, and the same bound holds. (ii) Outside the disc the weight is $\propto e^{-|k_z|h}/|k_z|$, integrable on every annulus and exponentially decaying; against the Dirichlet kernels, whose tails decay as $1/k$ per axis, the product is absolutely integrable for $h>0$, the dominated-convergence argument of the interior region applies to the whole plane, and the evanescent diagonal converges to $\hat x_p(\bk_i;h)$. (iii) is \cite[Lemma~2]{Iacovelli2026MC} and the discussion following it. \hfill$\blacksquare$

\section{Proof of Corollary~\ref{cor:ladder}}\label{app:ladder}
With $u_l=e^{-\im k_{z,i}z_l}$ and $v_l=e^{\im k_{z,i}z_l}$, (\ref{eq:tline}) reads $[\bar\bLam_i^{\rm ss}]_{ll'}=\frac{Z_{c,i}}{2}u_{\min(l,l')}v_{\max(l,l')}$, the semiseparable form of a transmission line seen at $L$ shunt nodes, and the inverse of a semiseparable matrix is tridiagonal. Computing it with the electrical lengths $\mu_l=k_{z,i}(z_{l+1}-z_l)$ of the gaps gives
\begin{align}
[(\bar\bLam_i^{\rm ss})^{-1}]_{l,l+1}&=\frac{-\im}{Z_{c,i}\sin\mu_l},\nonumber\\
[(\bar\bLam_i^{\rm ss})^{-1}]_{ll}&=\frac{\im\,(\cot\mu_{l-1}+\cot\mu_l)}{Z_{c,i}},\nonumber
\end{align}
where the cotangent missing at $l=1$ and at $l=L$ is replaced by $-\im$, the half-infinite sections below the first sheet and above the last being matched. These are the entries of the nodal admittance matrix of $L$ nodes joined by line sections, so the inverse fails to exist only if $\sin\mu_l=0$ for some gap, which cannot happen for an evanescent mode, $k_{z,i}$ being imaginary there and $\sin\mu_l$ a hyperbolic sine. The transmit column continues the same pattern, whence $(\bar\bLam_i^{\rm ss})^{-1}\bar\bLam_i^{\rm s0}=e^{\im k_{z,i}(z_1-z_0)}\be_1$: seen from the nodes, the transmitter is a source injected at the first sheet alone. Multiplying $(\bar\bLam_i^{\rm ss}+Z_s\bI_L+\hat\bM_i)\bq_{{\rm s},(i)}=-\bar\bLam_i^{\rm s0}q_{0i}$ by $(\bar\bLam_i^{\rm ss})^{-1}$ therefore gives the stated system, whose matrix is tridiagonal plus $(\bar\bLam_i^{\rm ss})^{-1}\hat\bM_i$; when the sheets are not interconnected $\hat\bM_i$ is diagonal, the sum stays tridiagonal, and the solve costs $O(L)$. \hfill$\blacksquare$

\section{Proof of Proposition~\ref{prop:hub}}\label{app:hub}
Stacking the codeword maps over the sheets, (\ref{eq:pullback3d}) gives $\bZ_{\rm s0}=\bPhi_{\rm s}^{\mathsf H}\bXi_{\rm s}\bLam^{(\rm s0)}\bXi\bPhi_0$ and (\ref{eq:codewordmaps}) gives $\bg_{k,\rm s}=\bPhi_{\rm s}^{\mathsf T}\bXi_{\rm s}\hat\bg^{(\rm s)}_k$, so that, with (\ref{eq:mixedblock}), $\bZ_{00}=\bF^{\mathsf H}\bLam^{(00)}\bF$, $\bZ_{0\rm s}=\bF^{\mathsf H}\bZ^{\rm m}_{0\rm s}$, and $\bZ_{\rm s0}=\bZ^{\rm m}_{\rm s0}\bF$: the transmit map stands on the right of every block that leaves the transmit plane and on the left of every block that reaches it. Hence $\bZ_{\rm s0}\biota_0=\bZ^{\rm m}_{\rm s0}\tilde\bq_k$, the induced currents of (\ref{eq:induced}) are $\biota_{\rm s}=-\bB^{\rm m}\tilde\bq_k$ with $\bB^{\rm m}=\bT\bZ^{\rm m}_{\rm s0}$, and substituting in (\ref{eq:effchannel}) gives the first identity of (\ref{eq:gm}). For the second, the input impedance factors through $\bF$ exactly, $\bZ_{\rm in}=\bF^{\mathsf H}[\bLam^{(00)}-\bZ^{\rm m}_{0\rm s}\bB^{\rm m}]\bF$ by (\ref{eq:Zin}), and $\bZ_{\rm in}$ is complex symmetric by Lemma~\ref{lem:passivity}, so $\Real\{\bZ_{\rm in}\}=\Herm\{\bZ_{\rm in}\}$ and the Hermitian form of $\bC^{\rm m}_{\rm in}$ at $\tilde\bq=\bF\bw$ is $\bw^{\mathsf H}\bC_{\rm in}\bw$, the power of Lemma~\ref{lem:passivity}. For the bound, fix the stack geometry and the loads and let $\bW$ be feasible for (P1) at any transmit configuration. Its modal image $\tilde\bq_k=\bF\bw_k$ is feasible for (P2), by the power identity just proved, and attains the same objective value, by the channel identity; (P2) maximizes the same objective over all of $\mathbb C^{N_{\Iset}}$, a superset of the codeword-generated set, and carries neither C2 nor C3. Enlarging the feasible set cannot lower the optimum. \hfill$\blacksquare$
\fi

\end{document}